%% file: main.tex
\documentclass[11pt, oneside]{article}
\usepackage[margin=1in]{geometry}
\usepackage[utf8]{inputenc}

\title{Minimax Rates for Robust Community Detection}
\author{Allen Liu \thanks{Email: \texttt{cliu568@mit.edu}. This work was supported in part by an NSF Graduate Research Fellowship, a Fannie and John Hertz Foundation Fellowship and Ankur Moitra's NSF CAREER Award CCF-1453261 and NSF Large CCF1565235.}\and Ankur Moitra \thanks{Email: \texttt{moitra@mit.edu}. This work was supported in part by a Microsoft Trustworthy AI Grant, NSF CAREER Award CCF-1453261, NSF Large CCF1565235, a David and Lucile Packard Fellowship and an ONR Young Investigator Award.}}
\date{\today}

\input{preamble.sty}

\begin{document}

\maketitle
\thispagestyle{empty}

\begin{abstract}
    In this work, we study the problem of community detection in the stochastic block model with adversarial node corruptions. Our main result is an efficient algorithm that can tolerate an $\epsilon$-fraction of corruptions and achieves error $O(\eps) + e^{-\frac{C}{2} (1 \pm o(1))}$
    where $C = (\sqrt{a} - \sqrt{b})^2$ is the signal-to-noise ratio and $a/n$ and $b/n$ are the inter-community and intra-community connection probabilities respectively. These bounds essentially match the minimax rates for the SBM without corruptions. We also give robust algorithms for $\mathbb{Z}_2$-synchronization. At the heart of our algorithm is a new semidefinite program that uses global information to robustly boost the accuracy of a rough clustering. 
        Moreover, we show that our algorithms are {\em doubly-robust} in the sense that they work in an even more challenging noise model that mixes adversarial corruptions with unbounded monotone changes, from the semi-random model. 
\end{abstract}

\clearpage
\pagenumbering{arabic} 
\section{Introduction}

%\subsection{Background}

The stochastic block model (SBM) was introduced by Holland, Laskey and Leinhardt \cite{holland1983stochastic} in 1983. It generates a random graph with a planted community structure. For now, we focus on the case of two equal-sized communities. The model works as follows: There is an unknown bisection of the $n$ nodes in the graph and each of the two groups is called a {\em community}. Pairs of nodes are connected by an edge independently according to the following rules: If the nodes belong to the same community, the connection probability is $a/n$. And otherwise, if the nodes belong to different communities, the connection probability is $b/n$. Here $a$ and $b$ are parameters and the goal is to understand how well we can approximate the planted bisection for different choices of $a$ and $b$. 

The stochastic block model has been extensively studied over the years. It is known that the model exhibits various sharp statistical phase transitions. For example, in the {\em weak recovery} problem, the goal is to get nontrivial agreement with the planted bisection. Decelle et al. \cite{decelle2011asymptotic} conjectured that weak recovery is possible iff $(a-b)^2 > 2 (a+b)$. This threshold is also called the Kesten-Stigum bound. Their conjecture was based on non-rigorous arguments from statistical physics. Mossel et al. \cite{mossel2018proof} and Massoulie \cite{massoulie2014community} proved the conjecture. Moreover they gave efficient algorithms that solve weak recovery down to the Kesten-Stigum bound. In the {\em exact recovery} problem, the goal is to recover the planted bisection exactly with high probability. Abbe et al. \cite{abbe2015exact} showed that exact recovery is possible iff $a = p \log n$ and $b = q \log n$ and $(\sqrt{p} - \sqrt{q})^2 > 2$. Hajek et al. \cite{hajek2016achieving} gave an efficient algorithm matching this bound based on semidefinite programming. Note that for exact recovery we need logarithmic average degree to preclude having isolated nodes. In contrast, weak recovery is possible with constant average degree. The results in this paper will work in both regimes. 

Recent works have focused on achieving the {\em minimax rates} for accuracy. In particular consider the quantity
$$\inf_{\widehat{\pi}} \sup_{\Theta} \mbox{err}(\widehat{\pi}, \pi)$$
Here $\widehat{\pi}$ is an estimator for the planted partition and $\Theta$ is a space of parameters for the stochastic block model. For example, we can consider the worst-case error over all stochastic block models on $n$ nodes with imbalance at most $\alpha$ and where the inter-community and intra-community connection probabilities are at least $a/n$ and at most $b/n$ respectively. Finally $\mbox{err}(\widehat{\pi}, \pi)$ denotes the fraction of misclassified nodes. It is impossible to determine which is the first community and which is the second community, so the error is only measured up to a global swap between the two communities. In particular $0 \leq \mbox{err} \leq 1/2$. In this notation, weak recovery is possible iff $\mbox{err}< 1/2$ independently of $n$ and exact recovery is possible iff $\mbox{err}= 0$ with high probability. And yet, studying the minimax rates allows us to ask sharper questions about the behavior of the optimal accuracy as a function of $a$ and $b$. 

Belief propagation is believed to obtain the optimal accuracy in a wide range of parameters. We can think about some of the conjectures, which are now theorems, as being pieces of this puzzle. For example, the trivial fixed points of belief propagation correspond to solutions that achieve $\mbox{err} = 1/2$. Decelle et al. \cite{decelle2011asymptotic} showed that when $(a-b)^2 > 2 (a+b)$ the trivial fixed point is unstable and thus it is natural to expect belief propagation to converge to another solution, which, presumably solves the weak recovery problem. Indeed the algorithms of Mossel et al. \cite{mossel2018proof} and Massoulie \cite{massoulie2014community} can be thought of as low-temperature limits of belief propagation. 
 In a remarkable work, Mossel, Neeman and Sly \cite{mossel2014belief} gave an algorithm that achieves the optimal error when the signal-to-noise ratio, defined as $C = (\sqrt{a} - \sqrt{b})^2$ is large enough\footnote{In their paper, they write the signal-to-noise ratio as $(a - b)^2/(2(a +b))$ which is always within a factor of $2$ of our definition.  They do not actually write their accuracy as an explicit function of $C$.  They only prove that the accuracy is the same as that achieved in a broadcast tree reconstruction problem as long as $C$ is at least some universal constant.} . Again their algorithm, particularly their method for boosting the overall accuracy of a rough initial estimate, has important parallels with belief propagation. Zhang and Zhou \cite{zhang2016minimax} gave an approximate characterization of the minimax rates. They showed that
$$\inf_{\widehat{\pi}} \sup_{\Theta} \mbox{err}(\widehat{\pi}, \pi) = e^{-\frac{C}{2} (1 \pm o(1))}$$
for the two-community, approximately balanced case. Here the $o(1)$ term is a function of $C$ that goes to zero as $C$ increases. They also prove generalizations to the imbalanced and $k$-community case, but their results are only information-theoretic and do not give any efficient algorithms.  Fei and Chen \cite{fei2020achieving} showed that the natural semidefinite program achieves this error exponent in the two community, balanced case. For more than two communities, there is believed to be a computational vs. statistical tradeoff beneath the Kesten-Stigum bound \cite{abbe2017community}, but nevertheless the natural conjecture is that belief propagation achieves optimal error among all computationally efficient estimators. %and moreover it does so even in the semi-random model \cite{blum1995coloring, feige2001heuristics} where a monotone adversary is allowed to add edges between nodes in the same community, and delete edges between nodes in different communities. Relatedly, Moitra et al. \cite{moitra2016robust} showed that the weak recovery problem becomes strictly harder in the semi-random model, and thus some low-order change in the error rates is reasonable to expect. 

\subsection{Our Results}

In this work, we ask an ambitious question: {\em Is it possible to compete with the minimax rates while being robust to adversarial corruptions?} We work in the node corruption model, where an adversary is allowed to arbitrarily control all the edges incident to an arbitrary $\eps$-fraction of the nodes in the graph (see Definition~\ref{def:corrSBM}). This models realistic settings where nodes represent agents who might make or break ties in a potentially malicious way so as to affect the outcome of a community detection/graph partitioning algorithm. Our main result is a positive answer to this question, along with a computationally efficient estimator for doing so:

\begin{theorem}\label{thm:main-inf}[Informal, see Theorem~\ref{thm:main-SBM1}]
There is a polynomial-time algorithm that given an $\eps$-corrupted BM with $n$ vertices, edge probabilities $b/n < a/n \leq 1/2$, and two communities of sizes between $\alpha n /2$ and $n/(2\alpha)$ for some constant $\alpha$, outputs a labelling that has expected error at most
\[
O(\eps) +  e^{-\frac{C}{2} (1 + o(1))} + o(1/n) 
\]
where $C = (\sqrt{a} - \sqrt{b})^2$.
\end{theorem}

Makarychev et al. \cite{makarychev2016learning} gave an algorithm for almost exact recovery that works even when an adversary may corrupt $o(n)$ edges.  However, their accuracy and robustness guarantees are weaker.  They require additional constraints on $\eps$ compared to $a,b$ \--- as we discuss in Section~\ref{sec:compare}, such constraints are actually necessary in the edge corruption model.  Stephan and Massoulie \cite{stephan2019robustness} studied node corruptions, but only allowed for $O(n^{\delta})$ nodes to be corrupted for some constant $\delta > 0$. Banks et al. \cite{banks2021local} and Ding et al. \cite{ding2022robust} studied weak recovery in the edge corruption model. In particular Ding et al. \cite{ding2022robust} showed that if $a$ and $b$ are above the Kesten-Stigum bound, there is some $\eps > 0$ for which weak recovery is still possible even when as many as $\epsilon n$ edges are adversarially added/deleted. Compared to our results, both the goal (competing with the minimax error vs. getting nontrivial error) and the corruption model (node corruptions vs. edge corruptions) are different. Moreover the order of quantifiers is different since in their work the fraction of corruptions is allowed to be an arbitrary function of $a$ and $b$. In contrast, our results give essentially tight bounds on the largest fraction of corruptions that can be tolerated while still competing with the minimax error in the stochastic model. See Section~\ref{sec:compare} for further discussion of node vs. edge corruptions and its effect on the minimax rates. Finally Acharya et al. \cite{acharya2021robust} studied the related problem of estimating the parameter $p$ in an Erdos-Renyi random graph $G(n, p)$ with node corruptions.

We also extend our results to the $k$ community case:

\begin{theorem}\label{thm:main-inf2}[Informal, see Theorem~\ref{thm:main-SBM2}]
There is a polynomial-time algorithm that given an $\eps$-corrupted  SBM with $n$ vertices, edge probabilities $b/n < a/n \leq 1/2$, and $k$ communities of sizes between $\alpha n /k$ and $n/(k\alpha)$ for some constants $k, \alpha$, outputs a labelling that has expected error at most
\[
O(\eps) +  e^{-\frac{\alpha C}{k}(1 + o(1)) } + o(1/n) 
\]
where $C = (\sqrt{a} - \sqrt{b})^2$.
\end{theorem}
\begin{remark}
Note that the exponent of $-\alpha C/k$ is also optimal, even in the non-robust setting, as it matches the lower bound proven in \cite{zhang2016minimax}.
\end{remark}

We also study the $\mathbb{Z}_2$ synchronization problem: There is an unknown vector $\ell \in \{\pm 1\}^n$ and we observe a spiked random matrix $$\frac{\lambda \ell \ell^T}{n} + \frac{W}{\sqrt{n}}$$ where $\lambda$ is a parameter and $W$ is a Gaussian Wigner matrix with iid entries that are mean zero, variance one Gaussians. The goal is to compute an estimate $\widehat{\ell}$ that minimizes the disagreement with $\ell$. Again we cannot determine the sign of $\ell$ so we measure disagreement between $\widehat{\ell}$ and $\ell$ with respect to a global sign flip. 

The analogues of many of the key results in community detection are known for $\mathbb{Z}_2$ synchronization too. For example, in the weak recovery problem the goal is to get an estimate $\widehat{\ell}$ that achieves non-trivial error, which is possible iff $\lambda > 1$ \cite{onatski2013asymptotic, perry2018optimality}. There are also sharp characterizations of the asymptotic mutual information \cite{deshpande2015asymptotic} which can be used to pin down the minimax rates. See also \cite{fei2020achieving}. Again we ask: is it possible to compete with the minimax rates in the presence of adversarial corruptions? In this setting we allow an adversary to arbitrarily corrupt the entries in an $\eps$-fraction of the rows/columns (see Definition~\ref{def:corrZ2}). Again, we show that this is possible, and give computationally efficient algorithms for doing so:

\begin{theorem}\label{thm:main-inf3}[Informal, see Theorem~\ref{thm:main-sync}]
There is a polynomial-time algorithm that given an $\eps$-corrupted $\Z_2$-synchronization instance with parameter $\lambda$, outputs a labelling that has expected error at most
\[
  O(\eps) + e^{-\frac{\lambda^2}{2}(1 + o(1)) }  + o(1/n) \,.
\]
\end{theorem}
\begin{remark}
Note that the exponent of $-\lambda^2/2$ is optimal, even in the non-robust setting, as it matches the lower bound proven in \cite{fei2020achieving}.
\end{remark}

These results come somewhat as a surprise. As we discussed earlier, the minimax rates are closely related to belief propagation, and belief propagation is inherently brittle, particularly to adversarial corruptions. Moreover, the positive results in the non-robust setting essentially all come from finding a good initial estimate of the planted bisection, and then boosting using some local procedure like having each node taking the majority vote of its neighbors \cite{mossel2014belief}. Even approaches based on semi-definite programming do this, within a primal-dual analysis \cite{fei2020achieving}. The main issue is that this approach is doomed in the setting of adversarial node corruptions. In particular, an adversary can game the algorithm in such a way that performing ``boosting" results in a new partition that only achieves trivial error (see Observation~\ref{obs:local}).

At the heart of our results is a way to perform robust boosting based on using global information about the entire graph. In particular, we give a semi-definite programming algorithm to discover and correct large sets of nodes that are unduly affecting the labels of many other nodes. More broadly, our work raises the exciting possibility that, maybe even beyond the stochastic block model, it is possible to compete with the sharp error rates achieved by belief propagation all while being provably robust. See Section~\ref{sec:tech-overview} for a more detailed technical overview. 

%\allen{Can we say like one sentence earlier that our results hold with semi-random noise and then add a pointer to here?}

\subsection{Doubly-Robust Community Detection}

Our work fits into the broader agenda of designing algorithms for inference and learning with strong provable robustness guarantees \cite{diakonikolas2019robust, lai2016agnostic, diakonikolas2017being, klivans2018efficient, diakonikolas2019sever, bakshi2020robust, chen2020online, hopkins2018mixture, liu2021settling, bakshi2020robustly, liu2021learning}. Much of the literature operates in a setting where samples are generated from a ``nice" distribution where the moments are regular and well-behaved and can be used to detect large groups of correlated outliers. So far, this is the case for our algorithms too, since we can rely on the predictable spectral properties of graphs generated from the stochastic block model. 

Without corruptions, the important work of Feige and Kilian \cite{feige2001heuristics} considered augmenting the stochastic block model with a monotone adversary. This is called the {\em semi-random model}. After a graph is sampled from the stochastic block model, but before it is revealed to our algorithm, the monotone adversary is allowed to arbitrarily add edges between pairs of nodes belonging to the same community, and delete edges between pairs of nodes belonging to different communities. This seemingly only makes the problem easier. But in fact designing algorithms that continue to work in the semi-random model is challenging and subtle. In many ways, the semi-random model prevents algorithms from being overtuned to the stochastic block model. For exact recovery, algorithms based on semidefinite programming continue to work in the semi-random model in the same range of parameters \cite{hajek2016achieving, perry2017semidefinite}. For weak recovery, Moitra et al. \cite{moitra2016robust} showed that it is no longer possible to get algorithms that work down to the Kesten-Stigum bound, and thus there is a strict information-theoretic separation between the stochastic and semi-random models. Finally Fei and Chen \cite{fei2020achieving} gave an algorithm, also based on semidefinite programming, that competes with the minimax error in the stochastic setting, even in the presence of a monotone adversary. 

Given that being robust is not just about tolerating adversarial corruptions, or any one single goal, it is natural to ask: {\em Are there doubly-robust algorithms for community detection?} In particular we want algorithms that work with both adversarial node corruptions and also an unbounded number of monotone changes. Indeed our main algorithms all extend to this challenging setting:

\begin{theorem}[Informal]
There is a polynomial time algorithm that given an $\epsilon$-corrupted semi-random SBM outputs a labelling that has the same expected error as in Theorems~\ref{thm:main-inf} and \ref{thm:main-inf2} subject to the same assumptions on the parameters. 
\end{theorem}

A major difficulty of working with node corruptions is that an adversary can corrupt much more than a linear number of edges, because he has control over all the edges incident to corrupted nodes. In the stochastic block model, there is a natural limit to how much an adversary would actually use this power because if some nodes are too high-degree they become easy to identify. But in a semi-random model, a monotone adversary can create many high-degree nodes that we would not actually want to delete. Thus the two types of adversaries can compound our difficulties, and create situations where an overwhelming majority of the edges have been corrupted but nevertheless we cannot easily remove them by deleting high-degree nodes. Our algorithms for $\mathbb{Z}_2$-synchronization are also doubly-robust as Theorem~\ref{thm:main-inf3} also continues to hold in an $\eps$-corrupted and semi-random model.

\subsection{Broader Context}

We briefly discuss some previous approaches to community detection and why they fail to obtain the robustness guarantees that we aim for here.  Some methods are based on computing statistics over non-backtracking walks, e.g. \cite{mossel2018proof}  or self-avoiding walks, e.g. \cite{hopkins2017bayesian}. When an adversary can control a constant fraction of the nodes or edges, he can force the expectation of these statistics to be incorrect so that they no longer are correlated with whether a pair of nodes is on the same side of the community or not. Other, related methods are based on spectral properties of the non-backtracking walk operator or a matrix counting all self-avoiding walks of a certain length, e.g. \cite{massoulie2014community}. The spectral properties of these matrices break down, and look fairly arbitrary, with corruptions. Moreover these techniques are for weak recovery, and one would need to boost to achieve optimal accuracy for larger signal-to-noise ratio. There are approaches for boosting based on belief belief propagation, e.g. \cite{mossel2014belief}.  However they are based on approximating the posterior distribution of the community labeling, and when there are corruptions there is no reason for the posterior to achieve optimal, or even non-trivial accuracy. There are also approaches for community detection that are based on semidefinite programming, e.g. \cite{feige2001heuristics, guedon2014community, abbe2015exact, hajek2016achieving, montanari2015semidefinite}. However, they are all based on an SDP relaxation for the minimum bisection problem, and when an adversary can control high-degree nodes, he can alter the minimum bisection so that it becomes essentially uncorrelated with the planted community structure.  While there are modifications such as in \cite{makarychev2016learning, fei2020achieving}  that can deal with corruptions, these modifications are not able to achieve the types of strong robustness and accuracy guarantees that we obtain here.

\input{problem-setup}

\input{technical-overview}

\section{Notation and Preliminaries}\label{sec:prelims}
\paragraph{Matrix and vector notation} For a matrix $M$, we use $\norm{M}_{\op}$ to denote its operator norm, $\norm{M}_F$ to denote its Frobenius norm and $\norm{M}_1$ to denote its trace norm.  We use $\mathbf{1} = (1, \dots , 1)^T$ to denote the all ones vector and $J$ to denote the all ones matrix.  The dimensions will always be clear from context.  Throughout this paper, we will view vectors in $\R^n$ as column vectors.
\\\\
Next, we have a few definitions that will be used repeatedly later on.
\begin{definition}\label{def:hadamard-product}
For matrices $A,B$ of the same size, we use $A \odot B$ to denote their Hadamard product.
\end{definition}

\begin{definition}\label{def:restrict-matrix}
For an $n \times n$ matrix $A$ and subsets $S,T \subset [n]$, we use $A_{S \times T}$ to denote the $|S| \times |T|$ matrix obtained by taking the submatrix of $A$ indexed by $S \times T$. 
\end{definition}

\begin{definition}
For $0 < p,q < 1$, define
\[
R(p,q) = \frac{p(1 - q)}{q(1 - p)} \,.
\]
\end{definition}
\begin{definition}
For $0 < p,q < 1$ with $p \neq q$, define
\[
D(p,q) = \frac{\log \frac{1 - q}{1 - p}}{\log \frac{p(1 - q)}{q(1 - p)}} \,.
\]
\end{definition}

We have the following useful bound on $D(p,q)$.

\begin{claim}\label{claim:log-bound}
For any $0 < p,q < 1$ with $p \neq q$, we have $q < D(p,q) < p$.
\end{claim}
\begin{proof}
See Appendix~\ref{appendix:prelims}.
\end{proof}

\input{basic-concentration}

\input{tail-bounds}

\input{initialization-SDP}

\input{boosting-SDP}

\input{boosting-SDP-kcommunity}

\bibliographystyle{alpha}
\bibliography{bibliography}

\input{appendix}

\end{document}

%% file: problem-setup.tex
\section{Problem Setup}

We now formally define the models and problems that we study.

\subsection{Community Detection}

We begin with a standard definition of a stochastic block model.

\begin{definition}[Stochastic Block Model (SBM)]
An SBM is a graph on $n$ nodes generated as follows.  There is some unknown partition of $[n]$ into $k$ sets $S_1, \dots , S_k$.  Given parameters $a,b$ with $a > b$, a graph is then generated where nodes in the same community are connected with probability $a/n$ and nodes in different communities are connected with probability $b/n$ (all edges are sampled independently).
\end{definition}

Next, we introduce the notion of semi-random noise, where an adversary can make arbitrarily many ``helpful" changes.  

\begin{definition}[Semi-Random  SBM]
A Semi-random SBM is a graph on $n$ nodes generated as follows.  There is some unknown partition of $[n]$ into $k$ sets $S_1, \dots , S_k$.  Given parameters $a,b$ with $a > b$, a graph is generated where nodes in the same community are connected with probability $a/n$ and nodes in different communities are connected with probability $b/n$ (all edges are sampled independently).  An adversary may then arbitrarily add additional edges within communities and remove edges between different communities.
\end{definition}

Finally, we introduce a corruption model where the adversary may completely corrupt an $\eps$-fraction of nodes and make semi-random changes on the rest of the graph.  Algorithms that work in this setting need to be doubly-robust to both the small fraction of corrupted nodes and the semi-random noise.

\begin{definition}\label{def:corrSBM}[$\eps$-Corrupted Semi-Random SBM]
A $\eps$-corrupted Semi-Random SBM is a graph on $n$ nodes generated as follows.  There is some unknown partition of $[n]$ into $k$ sets $S_1, \dots , S_k$.  Given parameters $a,b$ with $a > b$, a graph is generated where nodes in the same community are connected with probability $a/n$ and nodes in different communities are connected with probability $b/n$ (all edges are sampled independently).  An adversary may then 
\begin{itemize}
    \item Arbitrarily add additional edges within communities and remove edges between different communities
    \item Pick up to $\eps n$ nodes and modify their incident edges arbitrarily
\end{itemize}
\end{definition}
\begin{remark}
In later sections, we will often just say $\eps$-corrupted SBM instead of $\eps$-Corrupted Semi-Random SBM but it will always refer to an SBM with both adversarial corruptions and semi-random noise.
\end{remark}

The goal of the learner is to observe a graph generated from an $\eps$-Corrupted Semi-random SBM and output a partition of $[n]$ that is close to the unknown partition.  Formally, if the learner outputs $S_1, \dots , S_k$ and the true partition is $\wt{S_1}, \dots , \wt{S_k}$ then the error is the minimum number of errors over all permutations of the sets i.e. 
\[
\text{err} = \min_{\pi:[k] \rightarrow [k]} \left( \sum_{i = 1}^k |S_i \backslash \wt{S_{\pi(i)}}| \right)\,.
\]
We will use the term accuracy for $1 - \text{err}$.

We will assume that the learner is given the parameters $a,b,k, \eps$ and also a parameter $\alpha$ such that $\alpha n/k \leq |S_i| \leq n/(\alpha k) $ for all $i$ (so $\alpha$ bounds the imbalance in the community sizes).  Note that even without corruptions (but with semi-randomness), there are known obstacles to recovering the planted partition without knowledge of the parameters \cite{perry2017semidefinite}.  Of course, if the parameters are unknown, we can simply guess them using a grid and output a list of candidate partitions at least one of which must have high accuracy.

\subsection{$\Z_2$-Synchronization}

Next, we define the problem of $\Z_2$-Synchronization.

\begin{definition}[$\Z_2$-Synchronization]
We are given an $n \times n$ matrix $A$ generated as follows.  There is an unknown sign vector $\ell \in \{-1,1 \}^n$.  We then observe $\lambda \ell \ell^T/\sqrt{n}  +  E$ where $\lambda$ is some parameter and $E$ has entries drawn i.i.d from $N(0,1)$.
\end{definition}

Similar to before for SBMs, we can define a natural extension of the above model that allows for semi-random noise.

\begin{definition}[Semi-random $\Z_2$-Synchronization]
We are given an $n \times n$ matrix $A$ generated as follows.  There is an unknown sign vector $\ell \in \{-1,1 \}^n$.  We then observe $\lambda \ell \ell^T/\sqrt{n}  +  E + F$ where $\lambda$ is some parameter, $E$ has entries drawn i.i.d from $N(0,1)$ and $F$ has the same signs as $\ell \ell^T$ (entrywise).
\end{definition}

Finally, we define a model that allows for an $\eps$-fraction of adversarial corruptions as well as semi-random noise.

\begin{definition}\label{def:corrZ2}[$\eps$-Corrupted Semi-random $\Z_2$-Synchronization]
We are given an $n \times n$ matrix $A$ generated as follows.  There is an unknown sign vector $\ell \in \{-1,1 \}^n$.  Let $E$ be an $n \times n$ matrix whose entries are drawn i.i.d from $N(0,1)$.  Let $A_0 = \lambda \ell \ell^T/\sqrt{n} + E$.  Now an adversary may modify $A_0$ by
\begin{itemize}
    \item Adding a matrix $F$ whose entries have the same signs are $\ell \ell^T$
    \item Picking up to $\eps n$ elements of $[n]$ and modifying the corresponding rows and columns arbitrarily
\end{itemize}
We observe the matrix $A$ after the adversary makes the above modifications to $A_0$.
\end{definition}

As usual, the goal of the learner is to observe $A$ generated as above and output a partition of $[n]$ that is close to the unknown partition given by the signs of $\ell$ where error is defined as the minimum disagreement up to flipping the components of the partition.  

%We assume that we are given $\eps, \lambda$.  %Do we need beta for imbalance?

\section{What is the Right Accuracy?}

Our goal is to give algorithms that achieve nearly optimal error in the presence of corruptions and semi-random noise.  We first discuss prior work that characterizes the optimal error in a non-robust setting i.e. without corruptions or semi-random noise. 

\subsection{Community Detection}

In \cite{zhang2016minimax}, the authors characterize the optimal accuracy achievable information-theoretically in a pure SBM (with no semi-random noise or corruptions) as the signal-to-noise ratio goes to infinity.  

\begin{theorem}[\cite{zhang2016minimax}]\label{thm:SBM-accuracy}
Consider a (pure) SBM on $n$ nodes with $k$ communities with edge probabilities $a/n$ and $b/n$.  Also assume that all communities have sizes between $\alpha n /k$ and $n/(\alpha k)$.  Assume that $a,b = o(n)$ and define $C = (\sqrt{a} - \sqrt{b})^2$.  Then as $C/(k \log k) \rightarrow \infty$ any algorithm must incur expected error at least
\[
 \begin{cases}
e^{-(1 + o(1)) \frac{C}{2} } \text{ if } k = 2 \\
e^{-(1 + o(1)) \frac{\alpha C}{k} } \text{ if } k \geq 3 
\end{cases}
\]
where the $o(1)$ is some quantity that goes to $0$ as $C/(k \log k) \rightarrow \infty$.
\end{theorem}

Note that the factor of $\alpha$ shows up only when $k \geq 3$ because then there can be two small communities of size $\alpha n /k$ but this cannot happen for $k = 2$.  In \cite{zhang2016minimax}, the authors also prove that the above is tight when $\alpha \geq \sqrt{3/5}$ in the sense that it is indeed possible to achieve the accuracy specified in Theorem~\ref{thm:SBM-accuracy}.  However, their proof is only information-theoretic and they do not give a polynomial time algorithm that achieves this.

In \cite{fei2020achieving}, the authors give a polynomial time algorithm for matching the accuracy in Theorem~\ref{thm:SBM-accuracy} for $k = 2$ and balanced communities.  Their algorithm is based on the max-cut SDP and also works in the semi-random model. 
\begin{theorem}[Informal \cite{fei2020achieving} ]
Consider a (pure or semi-random) SBM on $n$ nodes with two balanced communities with edge probabilities $a/n$ and $b/n$.   Assume that $a,b = o(n)$ and define $C = (\sqrt{a} - \sqrt{b})^2$.  There is an algorithm that achieves error $e^{-C/2+  O(\sqrt{C})}$ with $1 - o(1)$ probability.
\end{theorem}
There are a few additional technical conditions in their theorem such as the fact that the implicit constant in the $O(\sqrt{C})$ may depend on the ratio $a/b$ (but for say fixed $a,b$, it is a universal constant).  Nevertheless, as far as we are aware, this is the best known explicit bound on the classification accuracy in an SBM in terms of the signal-to-noise ratio.

Our main theorems, stated below, essentially match this guarantee but are significantly more general \--- they work for imbalanced and more than two communities and in the presence of a $\eps$-fraction of adversarial corruptions.
\begin{theorem}[Robust Community Detection with $k = 2$]\label{thm:main-SBM1}
There is a polynomial-time algorithm (Algorithm~\ref{alg:full-k=2}) that when run on an $\eps$-corrupted semi-random SBM with $n$ vertices, edge probabilities $b/n < a/n \leq 1/2$ and $k = 2$ communities of sizes between $\alpha n /2$ and $n/(2\alpha)$, outputs a labelling that has expected error at most
\[
O(\eps \alpha^{-3}) +  e^{-C/2  + O(\alpha^{-9} \sqrt{C} \log C)} + \frac{e^{-\sqrt{\log n}}}{n} 
\]
where $C = (\sqrt{a} - \sqrt{b})^2$.
\end{theorem}

\begin{theorem}[Robust Community Detection with $k \geq 3$]\label{thm:main-SBM2}
There is a polynomial-time algorithm (Algorithm~\ref{alg:full-general}) that when run on an $\eps$-corrupted semi-random SBM with $n$ vertices , edge probabilities $b/n < a/n \leq 1/2$ and $k$ communities of sizes between $\alpha n /k$ and $n/(k\alpha)$, outputs a labelling that has expected error at most
\[
O(\eps k/\alpha^3) +  e^{-\alpha C/k  + \poly(k/\alpha) \sqrt{C} \log C} + \frac{e^{-\sqrt{\log n}}}{n} 
\]
where $C = (\sqrt{a} - \sqrt{b})^2$.
\end{theorem}

In Theorems~\ref{thm:main-SBM1} and \ref{thm:main-SBM2}, the hidden constants in the $O( \cdot )$ and $\poly( \cdot )$ are all universal constants.  We imagine that $\alpha, k$ are fixed constants and that $C$ is sufficiently large as a function of $\alpha , k$.  We take $n \rightarrow \infty$.  $C$ may be held constant as $n$ grows or it may grow with $n$.  Our error can be decomposed as follows.  The $O(\eps)$ term comes from the corruptions.  The exponential term comes from the error that must be incurred, even without any corruptions.  Note that the leading terms in the exponents in our error guarantees \--- $-C/2$ for $k = 2$ and $-\alpha C/k$ for $k \geq 3$ \--- are sharp in that they exactly match those in Theorem~\ref{thm:SBM-accuracy}.  The last term is $o(1/n)$ so it contributes no additional errors with high probability.

\subsection{$\Z_2$-Synchronization}

The paper \cite{fei2020achieving} also gives recovery guarantees for $\Z_2$-synchronization based on the max-cut SDP.  They also prove a lower-bound that nearly matches their recovery guarantee.

\begin{theorem}[Informal \cite{fei2020achieving} ]\label{thm:Z2-accuracy}
Consider a (pure or semi-random) $\Z_2$-synchronization instance on $n$ variables with parameter $\lambda$.  Then there is an algorithm that achieves error $e^{-\lambda^2/2 + O(\lambda)}$ with $1 - o(1)$ probability.  Furthermore, no algorithm can achieve error better than $e^{-(1 + o(1))\lambda^2/2}$. 
\end{theorem}

Again, for $\Z_2$-synchronization, we are able to essentially match this guarantee, but robustly, in the presence of an $\eps$-fraction of adversarial corruptions.

\begin{theorem}[Robust $\Z_2$-Synchronization]\label{thm:main-sync}
There is a polynomial-time algorithm (Algorithm~\ref{alg:full-Z2}) that when run on an $\eps$-corrupted semi-random $\Z_2$-synchronization instance with parameter $\lambda$, outputs a labelling that has expected error at most
\[
  O(\eps) + e^{-\lambda^2/2  + O(\lambda)}  + \frac{e^{-\sqrt{\log n}}}{n} \,.
\]
\end{theorem}

As before, all of the hidden constants in the $O( \cdot )$ are universal constants.  We assume that $\lambda$ is at least some sufficiently large universal constant.  We imagine taking $n \rightarrow \infty$ and $\lambda$ may be held constant or it may grow with $n$.  As before, our error guarantee can be decomposed into the $O(\eps)$ term for the corruptions and the exponential term for the error that must be incurred even without any corruptions.  The leading term of $-\lambda^2/2$ in the exponent is sharp as it matches the lower bound in Theorem~\ref{thm:Z2-accuracy}.

%% file: technical-overview.tex
\section{Technical Overview}\label{sec:tech-overview}

For simplicity, consider the case $k = 2$ and assume that the communities are balanced.  Let the inter-community and intra-community connection probabilities be $a/n$ and $b/n$ respectively and define $C = (\sqrt{a} - \sqrt{b})^2$.  A standard approach to achieving strong accuracy guarantees in community detection is to first obtain a rough labelling and then boost it. In our case, for the rough labelling, there is some additional work required to achieve robustness to node corruptions (see Section~\ref{sec:initialization}). However the boosting step is the key component and we will focus on that here.  A natural first attempt would be to simply boost the accuracy with majority voting i.e. we set each node's label to agree with the majority of its neighbors.  Without corruptions, it turns out that if we have a rough clustering with $1/\sqrt{C}$ error (which our rough clustering algorithm in Section~\ref{sec:initialization} does obtain) then one step of boosting will already get to $e^{-C/2 + O(\sqrt{C})}$ error.  However, corruptions make the problem significantly more difficult.  In particular it turns out that the adversary can make it so that majority voting gets most of the labels wrong, even while keeping the degrees of the nodes fixed:
\begin{observation}[Informal]\label{obs:local}
Consider an SBM with edge probabilities $a/n,b/n$ and assume that we are given the true labelling.  An adversary can choose $\eps n$ nodes and corrupt them while preserving their degrees so that majority voting gets \emph{most labels wrong} as long as $\eps \geq \frac{2(a-b)}{a + b}$.
\end{observation}
The adversary can accomplish this by picking $\eps n$ nodes and reconnecting all of their edges to random nodes in the opposite community.  Now, on average, nodes have $(1 - \eps)a/2$ neighbors within their community and $(1 - \eps)b/2 + \eps(a + b)/2$ neighbors in the opposite community and the latter quantity is larger by the definition of $\eps$.

This precludes obstacle classes of ``local" algorithms that attempt to label each node based solely on information within its neighborhood.  In particular, such algorithms cannot get guarantees that depend only on the signal-to-noise ratio $C$ because, in the argument above, by increasing the degrees $a,b$  while holding $C$ fixed, the fraction of corruptions that can be tolerated goes to $0$.  Thus we will need to exploit global information in order to get a stronger and more robust boosting algorithm.

At a high-level we will design a suitable ``stability" property with respect to a graph (given by its adjacency matrix $A$) and a labelling of its vertices given by $\ell \in \{-1, 1\}^n$. It will have the key properties that
\begin{itemize}
    \item The subgraph of uncorrupted nodes with correct labels will satisfy this stability property
    \item A subgraph with too many mislabeled nodes will violate this stability property
\end{itemize}
%Furthermore, our stability property will be resilient to outliers in that a small fraction of corrupted nodes cannot hide an otherwise unstable subgraph.  This will ensure that even with corruptions, stable subgraphs cannot contain many mislabeled nodes.  
Our algorithm will proceed as follows: we find the largest subgraph for which the stability property is satisfied.  We argue that this must give us essentially only the correctly labeled  nodes.  We then flip the labels on all of the nodes outside this subgraph and argue that this must boost the accuracy because almost all of those nodes must have been mislabeled to begin with.   

\paragraph{Stability} Now we discuss the stability property that we use in more detail.  First, consider an SBM (i.e. with no corruptions) and let $A$ be its adjacency matrix.  Let\footnote{For technical reasons, in the proof we will use a slightly different adjustment to demean $\wh{A}$ but the definition here suffices for the purposes of the overview}
\begin{equation}\label{eq:demean}
\wh{A} = A - \frac{a + b}{2} J \,.
\end{equation}
Let the true labelling of the vertices be given by a vector $\wt{\ell} \in \{-1 , 1\}^n$ and define the matrix $\wt{L} = \wt{\ell} \wt{\ell^T}$.  The first key observation is that the entries of $\wh{A} \odot \wt{L}$ have positive expectation, where $\odot$ denotes the Hadamard product.  It can be shown that for any $\gamma n \times (1-10\gamma)n$ combinatorial rectangle with $\gamma \geq e^{-C + O(\sqrt{C})}$, the sum of the entries of $\wh{A} \odot \wt{L}$ is positive with high probability.  In particular the lower bound on $\gamma$ comes from tail bounds on the binomial distribution i.e. $\sim e^{-C + O(\sqrt{C})}$ fraction of the rows may actually have negative sum.  This lower bound on $\gamma$ is also exactly what shows up in our accuracy guarantee.  For the purposes of this overview, we will say that a matrix $\wh{A}$ is $\gamma$-stable with respect to a labelling $\wt{L}$ if the sum of the entries of $\wh{A} \odot \wt{L}$ over any $\gamma n \times (1-10\gamma)n$ combinatorial rectangle is positive.

Now we have a global notion of stability, involving subsets of $\gamma n$ rows, that we know $\wh{A}$ must satisfy.  This notion of stability is also robust to a small fraction of corruptions.  In particular, imagine that some matrix $\wh{A}$ is stable with respect to the labelling $\wt{L}$ where some of the nodes in $\wh{A}$ may be corrupted.  Because we enforce that the sum over any $\gamma n \times (1-10\gamma)n$ combinatorial rectangle of $\wh{A} \odot \wt{L}$ is positive, we can consider rectangles whose rows and columns correspond only to the uncorrupted nodes.  Thus, the submatrix of uncorrupted nodes must be stable as well (with a slightly worse parameter $\gamma$).  In other words, a small fraction of corruptions cannot hide an unstable submatrix \--- this is the key for dealing with adversarial corruptions.

\paragraph{Boosting Program} We can now describe our boosting procedure.  We take as inputs an adjacency matrix $\wh{A}$ and a rough labelling $\ell \in \{-1,1\}^n$.  Let $\theta$ be the error of $\ell$ with respect to the true labelling.  We may assume that $\theta$ is significantly larger than $\gamma$ and $\eps$ (the fraction of corruptions), since otherwise we can just terminate and output $\ell$.  Consider the following program:
\begin{definition}[Boosting Program (Informal)]\label{def:boostingSDP-informal}
Assume that we have a rough labelling of the vertices given by a vector $\ell \in \{-1 , 1\}^n$ and define the matrix $L = \ell \ell^T$.  We solve for a weight vector $w \in \R^n$ with entries between $0$ and $1$ such that  
\begin{itemize}
    \item In the matrix $\wh{A} \odot L \odot ( (\one  - w) (\one - w)^T)$, the sum of the entries over any $\theta n \times (1 - 10\theta) n$ combinatorial rectangle is positive 
\end{itemize}
Our objective is to minimize $\sum_{i} w_i$ .
\end{definition}
The above program is nonconvex, and ultimately we will need to work with a convex relaxation. The key step in the relaxation is to replace $ (\one  - w) (\one - w)^T$ with a matrix 
\begin{equation}\label{eq:weight-relaxation}
W = J - \begin{bmatrix}w_1 \one & \dots & w_n \one \end{bmatrix} - \begin{bmatrix}w_1 \one^T \\ \vdots \\ w_n \one^T \end{bmatrix} + N
\end{equation}
with a constraint on trace norm $\norm{N}_1$.  We also relax the notion of a combinatorial rectangle, replacing it with a matrix $N$ and a constraint on $\norm{N}_1$.  Note that trace norm constraints make intuitive sense for these relaxations because combinatorial rectangles are rank-$1$ and the trace norm is a useful convex relaxation for rank constraints.  There are several additional technical details that we need to deal with and constraints that we need to enforce but we will not discuss those here.  See Section~\ref{sec:relaxations} for the precise relaxed constraints and Section~\ref{sec:2community-SDP} for the full semidefinite progamming relaxation.    
 
The intended solution to the boosting program is $w = w_{\text{base}}$ where $w_{\text{base}}$ is the vector with $1$s in entries corresponding to corrupted or mislabeled nodes and $0$s in entries corresponding to uncorrupted and correctly labeled nodes.  Once we obtain a solution $w$, we simply flip the labels on all of the vertices where $w$ has large weight.  If $w$ is close to $w_{\text{correct}}$, then this procedure will boost the accuracy because we are correcting essentially all of the mistakes in the original labelling.  

\paragraph{Analysis} It remains to argue that the solution that we obtain must be close to $w_{\text{base}}$.  We do this in two steps.  We first prove that the intended solution $w_{\text{base}}$ is indeed feasible.  Informally (for the formal version see Lemma~\ref{lem:exists-solution}):
\begin{lemma}[Informal]\label{lem:informal-exists}
$w_{\text{base}}$ is a feasible solution to the boosting program
\end{lemma}

For the non-convex boosting program in Definition~\ref{def:boostingSDP-informal}, to argue that $w_{\text{base}}$ is feasible, it suffices to consider $\wh{A}$, constructed from \eqref{eq:demean} for an SBM (since $(\one - w_{\text{base}})(\one - w_{\text{base}})^T$ zeros out all corrupted nodes), and prove that it has nonnegative sums over all $\theta n \times (1 - 10\theta) n$ combinatorial rectangles.  Arguing about the relaxation requires more work.  The high level idea is to combine the trace constraints (for relaxed combinatorial rectangles) with spectral bounds on $A$ to bound their inner product.  In the constant-degree regime, $A$ does not actually satisfy the necessary spectral bounds (because nodes may have logarithmic degree).  Fortunately, we are able to appeal to results in \cite{chin2015stochastic} showing there is a large submatrix of $A$, corresponding to all nodes whose degree is not too high, that does satisfy the necessary spectral bounds.  This spectral bound on a large submatrix of $A$ turns out to suffice for our argument.

The second step in the analysis involves characterizing the structure of the optimal solution.  Informally (for the formal version see Lemma~\ref{lem:progress}): 
\begin{lemma}[Informal]\label{lem:informal-progress}
Any feasible solution to the boosting program with objective value at most $2\theta$ must satisfy that there are at most $0.1\theta$  nodes $i \in [n]$ that are mislabeled by $\ell$ and have $w_i \leq 0.9$.
\end{lemma}

The key intuition is that any solution must place essentially no weight on the mislabeled nodes because otherwise the combinatorial rectangle sum constraints of the program will be violated.  In particular, let $S$ be the set of mislabeled, uncorrupted nodes and $R$ be the set of correctly labeled, uncorrupted nodes.  Then  then the matrix $\wh{A} \odot L$ has entries with negative mean on the set indexed by $S \times R$.  Thus, if $w$ is not close to $1$ on the mislabeled nodes, then the sum of the entries of 
\[
\wh{A} \odot L \odot ((\one - w)(\one - w)^T)
\]
will be negative, violating the constraint of the program.  The full proof requires significantly more care because we replace $(\one - w)(\one - w)^T$ with a matrix $W$ with certain linear and trace norm constraints (recall \eqref{eq:weight-relaxation}).  Nevertheless, we use the decomposition of $W$ in \eqref{eq:weight-relaxation} and show that, combined with the spectral properties of $A$, an analog of the above argument can be pushed through.

Finally, it remains to combine Lemma~\ref{lem:informal-exists} and Lemma~\ref{lem:informal-progress} to complete the analysis of our boosting procedure.  The key point is that the two lemmas together imply that the optimal solution (where recall we are minimizing the total weight) must be close to $w_{\text{base}}$ because 
\begin{itemize}
    \item $w_{\text{base}}$ is feasible and places weight only on mislabeled and corrupted nodes
    \item  Our solution $w$ must place essentially full weight on all of the mislabeled nodes 
\end{itemize}
Thus, almost all of the nodes with low weight in the optimal solution must be mislabeled and so flipping their labels will actually reduce the error by a constant factor.  Ideally, we might hope that the boosting procedure achieves the optimal error in one shot.  However, this is not the case because there can be a constant fraction (e.g. $0.1$) difference between $w$ and $w_{\text{base}}$.  Still, we can then simply iterate this boosting step, re-solving the program again with the new labelling.  We show that we can keep iterating until we get down to the optimal error of $\sim \gamma + \eps$.

\paragraph{Semi-Random Noise} Working in the semi-random model can foil many natural algorithms.  In particular, it is no longer possible to enforce the same degree constraints or spectral constraints on the adjacency matrix.  For our analysis, however, dealing with semi-random noise follows almost immediately.  The crucial property about the way we formulated the boosting program is that we know the signs of the entries of all of the matrices that appear, namely $((\one - w)(\one-w)^T)$ and $L$.  Thus, we can explicitly reason about locations where the semi-random noise is positive or negative.

More specifically, recall that for Lemma~\ref{lem:informal-exists}, we consider $w = w_{\text{base}}$ which is the indicator vector of the corrupted and mislabeled nodes.  Thus, if we let $F$ denote the matrix of semi-random changes i.e. entries of $F$ are $1$ if the corresponding edge was added, $-1$ if the edge was removed, and $0$ otherwise , then
\[
F \odot L \odot ((\one - w)(\one - w)^T)
\]
is entry-wise nonnegative.  The monotone changes only strengthen the inequalities that our program enforces.  For Lemma~\ref{lem:informal-progress}, we restrict to a rectangle $S \times R$ where the nodes in $S$ are incorrectly labeled and the nodes in $R$ are correctly labeled.  Thus,
\[
F \odot L
\]
is entrywise nonpositive on $S \times R$.  Since  $((\one - w)(\one-w)^T)$ is entrywise nonnegative, if a constraint involving the rectangle $S \times R$ were violated before, the monotone changes only make the violated constraints worse, and consequently the same argument still works.

\subsection{What Happens with Edge Corruptions?}\label{sec:compare}

Several other papers \cite{makarychev2016learning, stephan2019robustness, banks2021local, ding2022robust} have considered a model where instead of node corruptions, we allow for edge corruptions, where an $\eps$-fraction of the edges may be corrupted.  However, there are several information-theoretic barriers to obtaining high accuracy with edge corruptions and thus the results in these papers necessarily have weaker guarantees.  First, note that
\begin{observation}
For $\eps \geq \frac{2(a - b)}{a + b}$, an SBM with edge probabilities $a/n, b/n$ where we an adversary may corrupt an  $\eps$-fraction of the edges is statistically indistinguishable from $G(n, (a+b)/(2n))$.
\end{observation}
The above holds simply because the adversary may delete edges within communities and add edges between different communities so that the effective edge probabilities are both $(a+b)/(2n)$.  Thus, we run into the same barrier as in Observation~\ref{obs:local} for local algorithms \--- it is impossible to obtain accuracy guarantees that depend only on the signal-to-noise ratio $C$ because the fraction of corruptions that can be tolerated goes to $0$ as $a,b$ increase while holding $C$ fixed.  

It is also worth noting that the edge-corruption model does not combine nicely with semi-random noise.  In particular, the fraction of edge corruptions can only scale with the number of edges in the pure SBM and not the total number of edges.  This is because the semi-random noise may simply add cliques on a small set of vertices.  In particular, in the sparse regime where $a,b$ are constant, the total number of edges is $\Theta(n)$.  If the semi-random noise adds a clique on $O(\sqrt{n})$ vertices, then the adversary is now allowed enough corruptions to erase the entire graph on the remaining vertices.
\begin{observation}
With semi-random noise added to an SBM with edge probabilities $a/n,b/n$, it is impossible to achieve accuracy better than $0.5$ if an adversary may corrupt an $\eps$-fraction of the total number of edges.
\end{observation}

\subsection{Paper Organization}

In Section~\ref{sec:prelims}, we introduce notation and prove a basic concentration inequality.  In Section~\ref{sec:tail-bounds}, we study the properties of the distribution of row sums in the adjacency matrix of an SBM.  In particular, we prove a generalized notion of the stability property discussed previously in the technical overview.  In Section~\ref{sec:initialization}, we present our initialization algorithm that computes a rough labelling.  Then in Section~\ref{sec:boost-2community}, we show how to robustly boost this rough labelling in the two-community case, by repeatedly solving an SDP that relaxes the combinatorial rectangle constraints discussed previously.  In Section~\ref{sec:boost-kcommunity}, we show how to generalize our robust boosting algorithm to the $k$-community case.  Roughly, we show how to enforce the constraints from the two-community SDP on all pairs of communities.  Finally, in Appendix~\ref{sec:Z2}, we prove our results for robust $\Z_2$-synchronization.  This follows almost exactly the same method and many of the pieces of the proof, such as the analysis of the boosting algorithm, can be directly re-used just with different parameter settings.

%% file: basic-concentration.tex
\subsection{Concentration Inequalities}

Now we need a concentration inequality for samples from a binomial distribution.  Intuitively, if we imagine that we have all of the labels of all but one vertex in an SBM with edge probabilities $a/n,b/n$, the inequalities below bound the probability that we can classify this vertex correctly.  As mentioned previously, this type of local boosting does not work in the presence of corruptions, but the bounds proved here will be used to prove more global properties that are used in our robust boosting procedure later on.

\begin{claim}\label{claim:binomials-imbalanced-diff}
Consider the distributions 
\begin{align*}
\mcl{D}_1 &= \textsf{Binom}(a/n, \alpha n ) - \textsf{Binom}(b/n, (1 - \alpha) n) \\
\mcl{D}_2 &= \textsf{Binom}(b/n, \alpha n ) - \textsf{Binom}(a/n, (1 - \alpha) n)
\end{align*}
where $0 < \alpha < 1$ and $b < a $.  Let $C = (\sqrt{a} - \sqrt{b})^2$.  Let 
\[
K = (1 - 2\alpha )n D(a/n, b/n)\,.
\]
Then for any $\theta$, we have
\[
\max\left( \Pr_{x \sim \mcl{D}_1}[x + K \leq \theta  ], \Pr_{x \sim \mcl{D}_2}[x + K \geq -\theta  ] \right) \leq  e^{-\frac{C}{2}  + \frac{\theta}{2} \log R(a/n,b/n)} 
\]
\end{claim}
\begin{proof}
See Appendix~\ref{appendix:prelims}.
\end{proof}

%We also have a slightly different concentration inequality for $\alpha \neq 1/2$ that instead compares the average of the two binomial distributions.
%\begin{claim}\label{claim:binomials-imbalanced-ratio}
%\end{claim}

%% file: tail-bounds.tex
\section{Key Properties of the SBM}\label{sec:tail-bounds}

In this section, we will leverage Claim~\ref{claim:binomials-imbalanced-diff} and a few other concentration inequalities to prove more global properties about the adjacency matrix of an SBM.  Roughly, if we let $\wh{A}$ be a centering of the adjacency matrix of a graph generated from an SBM i.e.
\[
\wh{A} = A - D(a/n,b/n)J 
\]
and $\ell \in \{-1,1 \}^n$ be the true labels, we will prove that most rows of $\wh{A} \odot (\ell \ell^T)$ have positive sum.  The way we formulate this is that for a matrix 
\[
Z = \begin{bmatrix} x_1 \mathbf{1}^T \\ \vdots \\ x_n  \mathbf{1}^T \end{bmatrix}
\]
with constant rows and $0 \leq x_1, \dots , x_n \leq 1$, we will lower bound $\langle \wh{A} \odot (\ell \ell^T),  Z \rangle $ in terms of the sum $x_1 + \dots + x_n$.  Formally, we define:
\begin{definition}\label{def:resolvable}
We say an $n \times n$ matrix $X$ is $(d_1, d_2)$-resolvable if for all $0 \leq x_1, \dots , x_n \leq 1$ with $x_1 + \dots + x_n \leq 10^{-6}n$, we have
\[
\left\la X ,  \begin{bmatrix} x_1 \mathbf{1}^T \\ \vdots \\ x_n  \mathbf{1}^T \end{bmatrix}  \right\ra \geq d_1 ( x_1 + \dots + x_n)  - d_2 n  \,. 
\]
\end{definition}
The above definition is useful because all of the expressions are linear in $x_1, \dots , x_n$ and thus once we prove that a matrix is resolvable, we don't need to, say, do casework on $x_1 + \dots + x_n$.

The main result that we will prove in this section is Corollary~\ref{coro:resolvable2}.  However, we will need several preliminary results.  First, we will need the following concentration inequality about the sums of the entries in combinatorial rectangles of a matrix with i.i.d. $0$-mean entries.
\begin{claim}\label{claim:removed-column-sums}
Let $n, n_1, n_2$ be a parameters with $n_1, n_2 \leq 10^{-6} n$.  Let $M$ be an $n \times n$ matrix whose entries are drawn from independent distributions with mean $0$, variance at most $\sigma^2$ for some $\sigma \geq 20/\sqrt{n}$ and always bounded between $\pm 1$.  Then with probability at least $1 -  e^{-8(n_1 + n_2) \log(n/n_1 + n/n_2)}  $, the magnitude of the sum of the entries over any $n_1 \times n_2$ combinatorial subrectangle of $M$ is at most $(n_1 + n_2) \sigma \sqrt{n}$.
\end{claim}
\begin{proof}
See Appendix~\ref{appendix:tail-bounds}.
\end{proof}

Now we begin to analyze row sums of the matrix $A - D(a/n,b/n)J$ where $A$ is the adjacency matrix of an SBM.  We first consider sets of rows of a fixed size (i.e. fixed $x_1 + \dots + x_n$) and then afterwards, we will aggregate over different possibilities for the value of $x_1 + \dots + x_n$.  In the claim below, $\kappa$ is a parameter that allows us to trade off how strong our bound is versus the failure probability.

\begin{claim}\label{claim:row-sum-distribution}
Let $A$ be the adjacency matrix of $SBM(a/n,b/n)$ where $b < a$ and let $\ell \in \{-1,1 \}^n$ be a vector that labels the two communities.  Let $C = (\sqrt{a} - \sqrt{b})^2$ and  $L = \ell \ell^T$.  Let $0 < \beta < 10^{-6}, \kappa \geq 1$ be some parameters.  Define
\[
Q = C/2 - \log(1 / \beta) - 3\kappa
\]
Then with probability at least $ 1 - e^{-\kappa \beta n} - 1/n^3$, for any $0 \leq x_1, \dots , x_n \leq 1$ with $x_1 + \dots + x_n = \beta n $ we have
\[
\left\la \left(A - D(a/n, b/n)J \right) \odot L, \begin{bmatrix} x_1 \mathbf{1}^T \\ \vdots \\ x_n  \mathbf{1}^T \end{bmatrix}  \right\ra \geq 2\beta n\left(\frac{Q}{\log R(a/n, b/n)} - \sqrt{a + b}\right) \,.
\]
\end{claim}
\begin{proof}
See Appendix~\ref{appendix:tail-bounds}.
\end{proof}

Now we can aggregate Claim~\ref{claim:row-sum-distribution} over different choices of $\beta$ to get an inequality where the RHS is linear in $x_1 + \dots + x_n$.  This will allow us to deduce resolvability of the matrix $(A - D(a/n,b/n)J) \odot L$.

\begin{lemma}\label{lem:row-sum-lowerbound}
Let $A$ be the adjacency matrix of $SBM(a/n,b/n)$ where $b < a$ and let $\ell \in \{-1,1 \}^n$ be a vector that labels the two communities.  Let $C = (\sqrt{a} - \sqrt{b})^2$ and $L = \ell \ell^T$. Let $\kappa , K$ be some parameters and assume $\kappa \geq 1, K , C \geq 10^4$.  Define
\[
\theta = e^{-C/2 + 3\kappa + K \sqrt{a + b} \log R(a/n,b/n) } 
\]
 Then with probability at least $1 - e^{-10 \kappa } - 1/n^2$, we have for any  $0 \leq x_1, \dots , x_n \leq 1$ with $x_1 + \dots + x_n \leq 10^{-6} n$ that
\[
\left\la \left(A - D(a/n, b/n)J \right) \odot L, \begin{bmatrix} x_1 \mathbf{1}^T \\ \vdots \\ x_n  \mathbf{1}^T \end{bmatrix}  \right\ra \geq  (K(x_1 + \dots + x_n) - \theta n )  \sqrt{a + b}  - \frac{\max(0, 10^4(\kappa -\sqrt{C}))}{\log R(a/n,b/n)}\,. 
\]
\end{lemma}
\begin{remark}
Note that the additive term in the expression on the RHS above kicks in only when $\kappa \geq \sqrt{C}$ i.e. when we want very low failure probability.  It is also only necessary in the proof when $x_1 + \dots + x_n \leq 100$ i.e. when we are only considering a very small number of rows.   We think of $K$ as tunable parameter that will be set later on.
\end{remark}
\begin{proof}
See Appendix~\ref{appendix:tail-bounds}.
\end{proof}

As an immediate corollary of Lemma~\ref{lem:row-sum-lowerbound}, we have the following statement about the resolvability of $(A - D(a/n,b/n)J) \odot L$.
\begin{corollary}\label{coro:resolvable1}
Let $A$ be the adjacency matrix of $SBM(a/n,b/n)$ where $b < a$ and let $\ell \in \{-1,1 \}^n$ be a vector that labels the two communities.  Let $C = (\sqrt{a} - \sqrt{b})^2$ and $L = \ell \ell^T$.   Let $\kappa , K$ be some parameters and assume $\kappa \geq 1, K , C \geq 10^4$.  Define
\[
\theta = e^{-C/2 + 3\kappa + K \sqrt{a + b} \log R(a/n,b/n) } 
\]
 Then with probability at least $1 - e^{-10 \kappa } - 1/n^2$, we have that the matrix $(A - D(a/n, b/n)J) \odot L$ is resolvable with parameters 
\[
\left( K\sqrt{a + b} , \theta \sqrt{a + b} + \frac{\max\left(0, 10^4(\kappa - \sqrt{C}) \right)}{n \log R(a/n,b/n)}\right) \,.
\]
\end{corollary}

However, we will actually need a stronger statement, namely that for all sufficiently large subsets $S \subset [n]$, the matrix $(A - D(a/n, b/n)J) \odot L$ is still resolvable even if we restrict to $S \times S$.

\begin{corollary}\label{coro:resolvable2}
Let $A$ be the adjacency matrix of $SBM(a/n,b/n)$ where $b < a$ and let $\ell \in \{-1,1 \}^n$ be a vector that labels the two communities.  Let $C = (\sqrt{a} - \sqrt{b})^2$ and $L = \ell \ell^T$.   Let $\kappa , K$ be some parameters and assume $\kappa \geq 1, K \geq 10^4, C \geq (10K)^2$.  Then with probability at least $1 - e^{-10\kappa} - 2/n^2$, for all subsets $S \subset [n]$ with $|S| \geq (1 - K/(10\sqrt{C}))n$ we have that the matrix $((A - D(a/n,b/n)J) \odot L)_{S \times S}$ is resolvable with parameters 
\[
\left( 0.5 K\sqrt{a + b} , 1.1\left( \left( \theta + \frac{n - |S|}{n}\right)\sqrt{a + b} + \frac{ \max(0, 10^4(\kappa - \sqrt{C}))}{n \log R(a/n,b/n)}\right)\right)
\]
where 
\[
\theta = e^{-C/2 + 3\kappa + K \sqrt{a + b} \log R(a/n,b/n) }  \,.
\]
\end{corollary}
\begin{proof}
See Appendix~\ref{appendix:tail-bounds}.
\end{proof}

\subsection{Spectral Bounds}

We will also need a few standard spectral bounds later on.  Recall that a common issue when dealing with sparse SBMs is that high-degree nodes can make spectral bounds worse by a $\log n$ factor than what one would ideally hope for.  The following result from \cite{chin2015stochastic} essentially says that if we delete the nodes whose degree is too high, then we no longer lose this $\log n$ factor.  Of course our algorithm cannot actually prune these nodes (since we need to deal with semi-random noise) but it suffices in our analysis to know that there exists a large submatrix of the pure SBM adjacency matrix that satisfies the spectral bounds without the $\log n$ factor.

\begin{theorem}[Spectral Bounds with Degree Pruning \cite{chin2015stochastic}]\label{thm:spectral-remove}
Suppose M is random symmetric matrix with zero on the diagonal whose entries above the diagonal are independent with the following distribution
\[
M_{ij} = \begin{cases} 1- p_{ij} &\text{ w.p. } p_{ij} \\ p_{ij} &\text{ w.p. } 1 - p_{ij}  \end{cases}
\]
Let $\sigma$ be a quantity such that $p_{ij} \leq \sigma^2$ and $M_1$ be the matrix obtained from M by zeroing out
all the rows and columns having more than $20 \sigma^2 n$ positive entries. Then with probability $1 - 1/n^2$, we have $\norm{M_1} \leq \chi \sigma \sqrt{n}$ for some universal constant $\chi$.
\end{theorem}

As a consequence of the above, we have the following spectral bound for the adjacency matrix of an SBM after pruning.

\begin{corollary}\label{coro:spectral-remove}
Let $A$ be the adjacency matrix of an SBM with communities $S_1, \dots , S_k \subset [n]$ and edge probabilities $a/n, b/n$ where $b < a$.  Let $L$ be the matrix that has $1$ in entries $L_{ij}$ where $i,j$ are in the same community and $-1$ in other entries.  Let $C = (\sqrt{a} - \sqrt{b})^2$.  Then with probability at least $1 - 2/n^2$, there is a subset $S \subset [n]$ of size at least $(1 - e^{-2C})n$ such that 
\[
\norm{\left(A - \frac{a + b}{2n}J - \frac{a - b}{2n}L \right)_{S \times S}}_{\op} \leq \chi \sqrt{a + b} 
\]
where $\chi$ is some universal constant.
\end{corollary}
\begin{proof}
We apply Theorem~\ref{thm:spectral-remove} with $\sigma = \sqrt{a/n}$.  Note that zeroing out the diagonal affects the operator norm by at most $1$.  It now suffices to bound the number of vertices with degree more than $20 a$.  Let $m = \lceil e^{-2C}n \rceil$.  The probability that there are at least $m$ vertices with degree at least $20a$ is at most 
\begin{align*}
\binom{n}{m} \binom{m n }{10am} \left(\frac{a}{n}\right)^{10am} \leq \left(\frac{en}{m} \right)^m \left( \frac{en}{10a}\right)^{10am}\left( \frac{a}{n}\right)^{10am} \leq e^{-10am + m (\log(n/m) + 1)} \leq \frac{1}{n^2}
\end{align*}
where in the above we simply union bound over possible choices of $m$ vertices and choices of edges so that each vertex has degree at least $20a$.  Union bounding the above failure probability with that in Theorem~\ref{thm:spectral-remove}, we are done.
\end{proof}

%% file: initialization-SDP.tex
\section{Initialization SDP}\label{sec:initialization}

Now we discuss the initialization step of our algorithm where the goal is to obtain a rough estimate of the labelling that has accuracy $1 - O(1/\sqrt{C})$.  The high-level idea is as follows.  If we consider the matrix $A - (a/n)J$ where $A$ is the adjacency matrix of a pure SBM, then the on-diagonal blocks corresponding to each community are spectrally bounded by $O(\sqrt{a + b})I$.  However, the whole matrix is not because the off-diagonal blocks corresponding to different communities have spectral norm $\sim (a-b)$.  We will solve for a weight matrix $W$ with entries between $0$ and $1$ such that 
\[
W \odot (A - (a/n)J)
\]
is spectrally bounded.  We argue that $W$ necessarily must place essentially all of its weight on on-diagonal blocks corresponding to the communities.  Thus, we simply solve for the $W$ that maximizes the sum of its entries and then post-process using $k$-means clustering to recover the communities.  This formulation already naturally deals with a small fraction of corrupted nodes because $W$ can simply be $0$ on the corresponding rows and columns. To deal with semi-random noise, we introduce an additional matrix $F$ with nonnegative entries and instead require that
\[
W \odot (A - (a/n)J - F)
\]
is spectrally bounded.  The main point is that the matrix $F$ can only make off-diagonal blocks ``worse" because those already have too few edges.  We now state the SDP formally.

\begin{definition}[Initialization SDP]\label{def:init-SDP}
Assume we are given some $n \times n$ matrix $A$ and parameters $a,b, \chi$.  Then we solve for $n \times n$ matrices $W, F$ such that 
\begin{align*}
&0 \leq W_{ij} , F_{ij}\leq 1 \quad \forall i,j \\
&\norm{W}_1 \leq n \\
& - \chi \sqrt{a + b} I \preceq (A  - (a/n)J - F) \odot W \preceq \chi \sqrt{a + b} I
\end{align*}
where $\chi$ is the same universal constant as in Corollary~\ref{coro:spectral-remove}.  The objective is to maximize $\sum_{i,j} W_{ij}$.
\end{definition}

The initialization algorithm is as we alluded to.  We simply solve the SDP and then run $k$-means clustering on the rows of the resulting solution.

\begin{algorithm}[H]
\caption{{\sc Compute Initial Labelling} }
\begin{algorithmic}
\State \textbf{Input:} Adjacency matrix $A$
\State Run Initialization SDP with $\chi$ set as in Corollary~\ref{coro:spectral-remove} to obtain solution $W$
\State Cluster rows of $W$ using $10$-approximate $k$-means algorithm (e.g. \cite{kanungo2004local})
\State \textbf{Output:} resulting clusters
\end{algorithmic}
\label{alg:round-initialSDP}
\end{algorithm}

The key lemma in the analysis of the initialization SDP is stated below.  We show that the constraints imply that the total weight of $W$ on entries $W_{ij}$ where $i$ and $j$ are in different communities must be very small.  We then construct a feasible solution that has essentially full weight on all other entries i.e. $W_{ij} = 1$ for $i,j$ in the same community, and thus the optimal solution must have this property as well.  

\begin{lemma}\label{lem:initial-SDP-analysis}
Let $A$ be the adjacency matrix of a $\eps$-corrupted SBM with edge probabilities $a/n,b/n$ and communities given by $S_1, \dots , S_k \subset [n]$ (where $S_1, \dots , S_k$ must partition $[n]$).  Let $C = (\sqrt{a} - \sqrt{b})^2$.  Then with probability at least $1 - 2/n^2$, the solution $W$ to the Initialization SDP  has the following properties:
\begin{itemize}
    \item The sum of the entries of $W_{ij}$ where $i,j$ are in different communities is at most $(5\chi/\sqrt{C} +2\eps ) n^2$
    \item The sum of the entries of $W_{ij}$ where $i,j$ are in the same community is at least \[\sum_{i = 1}^k|S_i|^2 - (6\chi /\sqrt{C}  + 4\eps)n^2 \,. \]
\end{itemize}
\end{lemma}
\begin{proof}
See Appendix~\ref{appendix:initialization}.
\end{proof}
We can now complete the analysis of our initialization algorithm by demonstrating that any $10$-approximate $k$-means clustering must essentially recover the true communities. 

\begin{lemma}\label{lem:rough-clustering}
Let $A$ be the adjacency matrix of a $\eps$-corrupted SBM with edge probabilities $a/n,b/n$ and communities given by $S_1, \dots , S_k \subset [n]$.  Assume that $\alpha n/k \leq |S_i| \leq n/(\alpha k) $ for all $i$.  Let $C = (\sqrt{a} - \sqrt{b})^2$.  Then with probability at least $1 - 2/n^2$, the output of {\sc Compute Initial Labelling}  has error at most
\[
\frac{10^4  k}{\alpha^3}\left( \frac{\chi}{\sqrt{C}} + \eps \right)
\]
where $\chi$ is a universal constant.
\end{lemma}
\begin{proof}
See Appendix~\ref{appendix:initialization}.
\end{proof}

%% file: boosting-SDP.tex
\section{The Robust Boosting Procedure}\label{sec:boost-2community}

Now we present our robust boosting algorithm.

\subsection{Relaxed Constraints}\label{sec:relaxations}

We begin by introducing relaxed notions of combinatorial rectangles \--- recall the non-convex program in Definition~\ref{def:boostingSDP-informal} that we will now replace with a convex one.

\begin{definition}[Pseudorectangles]
For an integer $n$ and parameter $0 < \theta < 1$, we define the class of $n \times n$ matrices, $\mcl{R}_n(\theta)$, which we call $\theta$-pseudorectangles, as follows: $M \in \mcl{R}_n(\theta)$ if
\begin{align*}
0 \leq M_{ij} \leq 1 \quad \forall i,j \\
\sum_{ij} |M_{ij}| \leq \theta^2 n^2 \\
\norm{M}_1 \leq  \theta n 
\,.
\end{align*}
\end{definition}

The above already gives a convex relaxation for combinatorial rectangles.  Note that $\mcl{R}_n(\theta)$ is a relaxed notion of rectangles of area $\theta^2 n^2$.  However, the above is not strong enough for very thin rectangles e.g. of dimensions $\theta n \times (1  - \theta)n$ for small $\theta$ and it is exactly these thin rectangles that will be crucial in our boosting program.  Note that for such thin rectangles, the trace constraint above becomes $\sqrt{\theta} n$ which is too weak.  However, to get around this, we can instead view such thin rectangles as a subset of $\theta n$ rows with a $\theta n \times \theta n$ rectangle subtracted off.  This motivates the following definition.

\begin{definition}[Approximate-row-selectors]
For an integer $n$ and parameters $0 < \theta, \delta < 1$, we define a $(\theta,\delta)$-approximate-row-selector as a matrix $M$ along with real numbers $x_1, \dots , x_n$ satisfying the following properties:  there is a matrix $N \in \mcl{R}_n(\sqrt{\theta \delta})$ such that 
\begin{align*}
&0 \leq M_{ij} \leq 1 \quad \forall i,j \\
&0 \leq x_1, \dots , x_n \leq 1 \\ 
&x_1 + \dots + x_n  \leq  \theta n \\ 
&M = \begin{bmatrix} x_1 \mathbf{1}^T \\ \vdots \\ x_n  \mathbf{1}^T \end{bmatrix}  - N \\ 
\end{align*}
We will say $(M, x_1, \dots , x_n) \in \mcl{S}_n(\theta, \delta)$ if the above conditions are satisfied.
\end{definition}
The intention is for an approximate-row-selector to select $\theta n$ rows with up to $\delta n $ columns removed.  This corresponds to when $N$ is the indicator of a $\theta n \times \delta n $ rectangle.  Observe that for $\theta = \delta$, the trace constraint on $N$ is now just $\theta n$ as opposed to $\sqrt{\theta} n$.

\subsection{Boosting SDP for $2$ Communities}\label{sec:2community-SDP}

Now we formulate the boosting SDP.  We first formulate the SDP for community detection with $2$ communities.  In the next section, Section~\ref{sec:boost-kcommunity}, we show how to extend the SDP to $k$ communities.
\begin{definition}\label{def:boostingSDP}[Boosting SDP For $2$ Communities]
Let $\wh{A}$ be an $n \times n$ matrix we are given as input.  Let $\ell \in \{-1, 1 \}^n$ be a vector of labels that we are given.  Let $\zeta, d, K$ be some parameters that we can set.  The boosting SDP for $\wh{A},\ell$ and parameters $\zeta , d,  K$ with $0 < \zeta < 1, K \geq 10^4$ is defined as follows: we have a variable $0 \leq \rho \leq \zeta$, weights $w_1, \dots , w_n \in \R$ and matrices $W \in \R^{n \times n}, N \in \mcl{R}_n(\rho)$ such that
\begin{align*}
&0 \leq w_1, \dots , w_n \leq 1 \\
&w_1 + \dots + w_n \leq \rho  n\\
&W_{ij} \geq 0 \quad \forall i,j \\
&W = J - \begin{bmatrix}w_1 \one & \dots & w_n \one \end{bmatrix} - \begin{bmatrix}w_1 \one^T \\ \vdots \\ w_n \one^T \end{bmatrix} + N
\end{align*}
and finally we have the constraint that for all $\rho'$ with $\rho/K \leq \rho' \leq \zeta $ and all $(M, x_1, \dots , x_n) \in \mcl{S}_n(\rho' , K \rho')$, we have 
\[
 \left\la \wh{A}  \odot L \odot W, M  \right\ra \geq  10d K^2\left( K(x_1(1 - w_1) + \dots + x_n(1 - w_n) ) -   \rho' n \right)  
\]
where $L = \ell \ell^T$.  The objective is to find a feasible solution that minimizes $\rho$.
\end{definition}
\begin{remark}
Technically, we should have the last constraint only be over $\theta$ that are integer multiples of $1/n^{10}$.  It will be obvious that all of our arguments work with this modification.  This also makes it clear that the above program can be optimized to say $1 - 1/n^{10}$ accuracy in polynomial time because we can binary search on $\rho$ and we have a separation oracle because we can optimize over the sets $\mcl{S}_n(\theta, \delta)$ by solving a SDP.
\end{remark}

In the boosting SDP, the matrix $\wh{A}$ is supposed to be the (appropriately de-meaned) adjacency matrix and the vector $\ell$ corresponds to the rough clustering that we will boost.  After solving the boosting SDP, we improve the accuracy of the labelling as described below.  Roughly, we just flip the labels on all nodes $i$ for which $w_i$ is large as these are the nodes that are ``down-weighted" by $W$.  We set $\zeta , d,  K$ in terms of the parameters of the SBM.  For the precise setting see Algorithm~\ref{alg:round-boostingSDP}.

\begin{algorithm}[H]
\caption{{\sc Boosting Using SDP} }
\begin{algorithmic}
\State \textbf{Input:} Matrix $\wh{A}$, parameters $\zeta , d , K$
\State \textbf{Input:} Rough labelling given by $\ell_{\text{init}} \in \{-1 , 1 \}^n$ 
\State Set $\ell \leftarrow \ell_{\text{init}}$
\For{$t = 1,2, \dots ,  10 \log n$}
\State Run Boosting SDP with $\ell, \wh{A}, \zeta,d , K$
\For {$i$ such that $w_i \geq 1 - 1/\sqrt{K}$}
\State Flip the label $\ell_i$ of $i$
\EndFor
\State Set $\ell$ to the new labelling (after flips)
\EndFor
\State \textbf{Output:} $\ell$
\end{algorithmic}
\label{alg:round-boostingSDP}
\end{algorithm}

We will first analyze the algorithm under certain deterministic conditions on the matrix $\wh{A}$ (i.e. we first abstract away the generative model).  Recall the definition of a matrix being resolvable (Definition~\ref{def:resolvable}).  The main theorem that we prove in this section is that under certain resolvability conditions on the input matrix, the algorithm {\sc Boosting Using SDP} will return a labelling with high accuracy.  Then in Section~\ref{sec:2community-full}, we combine this with results in Section~\ref{sec:tail-bounds} and Section~\ref{sec:initialization} to give a full algorithm for recovering the communities in an $\eps$-corrupted SBM.

\begin{theorem}\label{thm:boosting-SDP}
Assume that we run {\sc Boosting Using SDP} and that the input satisfies that $K \geq 10^4$ and there is a subset $S \subset [n]$ with $|S| \geq (1 - \gamma) n $ where $\gamma \leq 0.1 \zeta$, an $n \times n$ matrix $F$ and a labelling $\wt{\ell} \in \{-1,1 \}^n$ where we define $\wt{L} = \wt{\ell} \wt{\ell}^T$ such that
\begin{itemize}
    \item $\ell_{\text{init}}$ agrees with $\wt{\ell}$ on at least $(1 - 0.1 \zeta)n$ entries
    \item $(F \odot \wt{L})_{S \times S}$ is entrywise nonnegative
    \item $(\wh{A} - F)_{S \times S} = Y + Z$ where $\norm{Y}_{\op} \leq Kd$ and $|Z_{ij}| \leq Kd/(\zeta n )$ for all $i,j$
    \item For all subsets $T \subset S$ with $|T| \geq (1 - \zeta)n $ we have that $\left((\wt{A} - F) \odot \wt{L} \right)_{T \times T}$ is resolvable with parameters 
    \[
    \left(10dK^3, 0.5 dK\left( \gamma + \frac{n - |T|}{n} \right)\right) \,.
    \]
\end{itemize}
Then the output of {\sc Boosting Using SDP} agrees with $\wt{\ell}$ on at least $(1 - 8\gamma)n$ entries.
\end{theorem}

To prove Theorem~\ref{thm:boosting-SDP}, we need to analyze the boosting SDP.  We prove two key lemmas.  The first, Lemma~\ref{lem:exists-solution} characterizes when there exists a solution with small objective value.  The second, Lemma~\ref{lem:progress} characterizes the structure of any solution with good objective value.  We then combine these lemmas to argue that our algorithm makes progress in boosting the accuracy in each iteration. 

\subsubsection{Exists Good Solution}

First, we specify conditions that guarantee the existence of a good solution to the boosting SDP.  While the results below are stated with no generative model, it will be useful to keep in mind that when $A$ is generated from an SBM, the intended solution is for the matrix $W$ to be $1$ on the square corresponding all of the correctly labeled, uncorrupted nodes and $0$ everywhere else and then $\rho$ will be equal to the fraction of corrupted and mislabeled nodes.  

\begin{lemma}\label{lem:exists-solution}
Consider solving the boosting SDP.  Assume that there is an $n \times n$ matrix $F$ and a subset $S \subset [n]$ with $|S| \geq (1 - \theta)n$ and $\theta \leq \zeta$ such that the following properties hold: 
\begin{itemize}
    \item $(F \odot L)_{S \times S}$ is entrywise nonnegative
    \item $(\wh{A} - F)_{S \times S}  = Y + Z$ where $\norm{Y}_{\op} \leq K d$ and $|Z_{ij}| \leq K d/(\zeta n)$ for all $i,j$
    \item $\left( (\wh{A} - F) \odot L\right)_{S \times S}$ is $(10 d K^3  , \theta d K )$-resolvable 
\end{itemize}
Then letting $W$ be the $n \times n$ matrix that is $1$ on the square corresponding to $S \times S$ and $0$ everywhere else is a feasible solution that attains objective value $\rho  = \theta$.
\end{lemma}

\begin{proof}

Note that the proposed $W$ satisfies the constraint on $W$ because we can set $w_i = 1$ for exactly the elements of $[n] \backslash S$.  We then just need to subtract off the matrix with ones in the square indexed by $([n] \backslash S) \times ([n] \backslash S)$.  Now we need to show that the last constraint is satisfied.  Fix a $\rho'$ with $\theta \leq \rho' \leq \zeta$.  Write 
\[
M = \begin{bmatrix} x_1 \mathbf{1}^T \\ \vdots \\ x_n  \mathbf{1}^T \end{bmatrix}  - N
\]
as guaranteed by the definition of $\mcl{S}_n(\rho', K\rho')$.  Next, since $W$ and $M$ are entrywise nonnegative, we have
\[
\langle F \odot L \odot W, M \rangle \geq 0 \,.
\]
Now let
\[
X = (\wh{A} - F) \odot W \,.
\]
Note that 
\begin{align*}
\left\la \left(\wh{A} - F\right)  \odot L \odot W, N \right\ra &= \left\la X \odot L, N \right\ra \\ & = \left\la Y \odot L , N \right\ra + \left\la Z  \odot L , N \right\ra \\ &\leq  K \sqrt{K} \rho' d n  + K^2  d \rho'^2 n/\zeta  \\  &\leq  2K^2  \rho'  d n
\end{align*}
where in the above we first used that the Schatten-$1$-norm of $L \odot  N$ is at most $\sqrt{K} \rho' n$ by the trace constraint on $N$ and for the second term, we used that the sum of the entries of $N$ is at most $K\rho'^2n^2$.  Next, we bound
\begin{align*}
\left\la \left(\wh{A} - F \right) \odot L \odot W , \begin{bmatrix} x_1 \mathbf{1}^T \\ \vdots \\ x_n  \mathbf{1}^T \end{bmatrix}  \right\ra  &= \left\la X \odot L , \begin{bmatrix} x_1 \mathbf{1}^T \\ \vdots \\ x_n  \mathbf{1}^T \end{bmatrix}  \right\ra \\ &\geq 10 dK^3 (x_1(1 - w_1) + \dots + x_n(1 - w_n) ) -  \theta d K |S| \\& \geq  10 dK^3(x_1(1 - w_1) + \dots + x_n(1 - w_n) ) - K^2 \rho'  d n
\end{align*}
where the relation above follows from the resolvability assumption and the fact that the $w_i$ are constructed so that
\[
\sum_{i = 1}^n x_i(1 - w_i) = \sum_{i \in S} x_i \,.
\]
Putting the above inequalities together, we conclude 
\[
 \left\la \left(\wh{A} - F\right)  \odot L \odot W, M  \right\ra \geq 10d K^2 ( K(x_1(1 - w_1) + \dots + x_n(1 - w_n) -  \rho' n) 
\]
and this completes the verification of feasibility.

\end{proof}

\subsubsection{Any Feasible Solution Makes Progress}

The next lemma will be crucial to showing that our boosting algorithm actually makes progress after each solve of the SDP i.e. each iteration of the for loop.  Again, it is stated with no generative model but for analyzing SBMs it will be used as follows.  We compare the rough labelling $\ell$ to the ground truth labelling $\wt{\ell}$.  Let the error be $\theta$.  We show that any solution to the SDP that is close to optimal (in terms of the value of $\rho$) must have $w_i \geq 1 - 1/\sqrt{K}$ on the vast majority of the mis-labeled vertices in $\ell$.  However, by Lemma~\ref{lem:exists-solution} there is a solution that essentially only has positive $w_i$ on these mis-labeled vertices, so any near-optimal solution cannot also have too many large $w_i$ on correctly labeled vertices.  Thus, if we blindly flip the labels on vertices with $w_i \geq 1 - 1/\sqrt{K}$, we will actually reduce the error by a constant factor.

\begin{lemma}\label{lem:progress}
Consider solving the boosting SDP.  Let $\wt{\ell} \in \{-1 ,1 \}^n$ be a sign vector that agrees with $\ell$ on $(1 - \theta)n$ entries with $\theta \leq \zeta$.  Let $\wt{L} = \wt{\ell}\wt{\ell}^T$.  Assume that for some $\gamma \leq \zeta$, there is a subset $S \subset [n]$ with $|S| \geq ( 1 - \gamma)n$ and an $n \times n$ matrix $F$ such that the following properties hold:
\begin{itemize}
    \item $(F \odot \wt{L})_{S \times S}$ is entrywise nonnegative
     \item $(\wh{A} - F)_{S \times S}  = Y + Z$ where $\norm{Y}_{\op} \leq Kd$ and $|Z_{ij}| \leq K d/(\zeta n )$ for all $i,j$
    \item $\left( (\wh{A} - F) \odot \wt{L} \right)_{S \times S}$ is $(10 dK^3, \gamma d K)$-resolvable
\end{itemize}
Then for any feasible solution to the boosting SDP with $\rho \leq \theta + \gamma$, there are at most $10(\theta + \gamma)n/\sqrt{K}$ elements $i \in S$ such that $w_i \leq 1 - 1/\sqrt{K}$ and $\ell_i \neq \wt{\ell}_i$. 
\end{lemma}

\begin{proof}
Let $T \subset S$ be the set of elements in question.  Assume for the sake of contradiction that $|T| \geq 10(\theta + \gamma)n/\sqrt{K}$.  Note that $|T| \leq \theta n \leq \zeta n$.  Also $|T| \geq (\theta + \gamma)n/K$ because $K \geq 10^4$.

Let $R \subset S$ be the set of $i \in S$ where $\ell$ and $\wt{\ell}$ agree.  Now let $M$ be the indicator function of $T \times R$.  We claim that 
\[
M \in \mcl{S}_n( |T|/n , K |T|/n) \,. 
\]
This is because we can choose $x_1, \dots , x_n$ such that $x_i = 1$ if and only if $i \in T$ and $x_i = 0$ otherwise.  We are then removing at $(\theta + \gamma) n$ columns corresponding to elements of $[n] \backslash R$.  Thus, we can plug this setting of $M, x_1, \dots , x_n$ into the last constraint of the SDP.  By construction
\[
x_1(1 - w_1) + \dots + x_n(1 - w_n) \geq \frac{|T|}{\sqrt{K}}
\]
so the constraint enforces that
\begin{equation}\label{eq:constraint}
\left\la \wh{A} \odot L \odot W ,  M \right\ra \geq 10d K^2\left( K(x_1(1 - w_1) + \dots + x_n(1 - w_n) ) -   |T| \right) \geq  0 \,.
\end{equation}
We will actually prove that if $|T| \geq 10(\theta + \gamma)n/\sqrt{K}$ then the above will be violated and this will complete the proof by contradiction.  First, note that since $\ell$ and $\wt{\ell}$ disagree on all elements of $T$ and agree on all elements of $R$, we have
\[
\la F \odot L \odot W , M \ra = - \la F \odot \wt{L} \odot W , M \ra \leq 0 \,.
\]
Next, consider a feasible solution with $\rho \leq \theta + \gamma$ and write
\[
W = J - \begin{bmatrix}w_1 \one & \dots & w_n \one \end{bmatrix} - \begin{bmatrix}w_1 \one^T \\ \vdots \\ w_n \one^T \end{bmatrix} + N \,.
\]
We must have $w_1 + \dots + w_n \leq (\theta + \gamma) n$ and $N \in \mcl{R}_n(\theta + \gamma)$.  We write
\begin{align*}
\left\la \left( \wh{A} - F\right) \odot L \odot  W,  M \right\ra &= \left\la  \left( \wh{A} - F \right) \odot L,  \left( J - \begin{bmatrix}w_1 \one^T \\ \vdots \\ w_n \one^T \end{bmatrix} \right) \odot M \right\ra \\ &\quad - \left\la \left( \wh{A} - F\right) \odot L \odot  \begin{bmatrix}w_1 \one & \dots & w_n \one \end{bmatrix},  M \right\ra \\ &\quad + \left\la \left(\wh{A} - F \right) \odot L \odot N , M \right\ra
\end{align*}
and bound each of the terms on the RHS individually.  Let $X$ be the matrix obtained by taking $\wh{A} - F$ and zeroing out entries outside of $S \times S$.  First we bound
\begin{align*}
\left\la \left( \wh{A} - F\right) \odot L \odot  \begin{bmatrix}w_1 \one & \dots & w_n \one \end{bmatrix},  M \right\ra &=  \left\la X , L \odot    \begin{bmatrix}w_1 \one & \dots & w_n \one \end{bmatrix} \odot  M \right\ra \\ &= \left\la Y , L \odot   \begin{bmatrix}w_1 \one & \dots & w_n \one \end{bmatrix} \odot  M \right\ra \\ &\quad + \left\la Z , L \odot   \begin{bmatrix}w_1 \one & \dots & w_n \one \end{bmatrix} \odot  M \right\ra \\ &\geq -Kd (\theta + \gamma) n  - Kd (\theta + \gamma)^2 n/\zeta \\ &\geq -3 K d(\theta + \gamma) n   
\end{align*}
where we used that by construction, the matrix $L \odot   \begin{bmatrix}w_1 \one & \dots & w_n \one \end{bmatrix} \odot  M$ has trace norm at most  $(\gamma + \theta ) n$ and the sum of the absolute values of its entries is at most $(\gamma + \theta )^2 n^2$.  Similarly, we get
\begin{align*}
\left\la \left(\wh{A} - F \right) \odot L \odot N , M \right\ra &= \left\la X,  L \odot N \odot M \right\ra \\ &\leq Kd (\theta + \gamma) n + Kd (\theta + \gamma)^2 n/\zeta \\ &\leq 3 K d(\theta + \gamma) n   \,.
\end{align*}

Now to obtain a contradiction, we will prove that the remaining term involving 
\[
J - \begin{bmatrix}w_1 \one^T \\ \vdots \\ w_n \one^T \end{bmatrix}
\]
is sufficiently negative.  Define the variables $t_1, \dots , t_n$ such that $t_i = 1$ if and only if $i \in T$.  Let $H$ be the matrix that is the indicator function of $T \times (S\backslash R)$.  We have    
\begin{align*}
\left\la  \left(\wh{A} - F\right) \odot L,  \left( J - \begin{bmatrix}w_1 \one^T \\ \vdots \\ w_n \one^T \end{bmatrix} \right) \odot M \right\ra  = \left\la  X \odot L,  \left( J - \begin{bmatrix}w_1 \one^T \\ \vdots \\ w_n \one^T \end{bmatrix} \right) \odot M \right\ra \\ =  \left\la  X ,  \left( J - \begin{bmatrix}w_1 \one^T \\ \vdots \\ w_n \one^T \end{bmatrix} \right) \odot (M - H) \odot L \right\ra  + \left\la  X ,  \left( J - \begin{bmatrix}w_1 \one^T \\ \vdots \\ w_n \one^T \end{bmatrix} \right) \odot H \odot L \right\ra  \,.
\end{align*}
Now note that 
\[
\left((M - H) \odot L\right)_{[n] \times S} = \left(-(M + H) \odot \wt{L}\right)_{[n] \times S} \,.
\]
Also, the second term is at most $3K d(\theta + \gamma) n $ by the same argument as previously.  Thus we have
\begin{align*}
\left\la  \left( \wh{A} - F \right) \odot L,  \left( J - \begin{bmatrix}w_1 \one^T \\ \vdots \\ w_n \one^T \end{bmatrix} \right) \odot M \right\ra &\leq  - \left\la  X ,  \left( J - \begin{bmatrix}w_1 \one^T \\ \vdots \\ w_n \one^T \end{bmatrix} \right) \odot (M + H) \odot \wt{L} \right\ra + 3 K d(\theta + \gamma) n  \\ &= -\left\la  X \odot \wt{L},   \begin{bmatrix}(1 - w_1)t_1 \one^T \\ \vdots \\ (1 - w_n)t_n \one^T \end{bmatrix}   \right\ra + 3 K d(\theta + \gamma) n  \\& \leq -10 dK^3\left( \sum_{i \in S} (1 - w_i)t_i \right) + \gamma dK n  + 3 K d(\theta + \gamma) n  \\  & \leq -10 d K^{2.5} |T| + 4 dK (\theta + \gamma) n \\& \leq -10 dK^2(\theta + \gamma) n 
\end{align*}
where in the above we used the resolvability assumption.  Putting all of the inequalities together, we conclude
\[
\left\la  \wh{A} \odot L \odot W ,  M \right\ra = \left\la  (\wh{A} - F) \odot L \odot W ,  M \right\ra  + \left\la  F \odot L \odot W ,  M \right\ra  < 0
\]
which contradicts \eqref{eq:constraint} and completes the proof.
\end{proof}

\subsubsection{Proof of Theorem~\ref{thm:boosting-SDP}}
To complete the proof of Theorem~\ref{thm:boosting-SDP}, we will simply prove that after each iteration of the for loop in {\sc Boosting using SDP}, if the error is at least $8\gamma$ to start with, then the error of the new labelling is a constant factor smaller.  Since we run a sufficient number of iterations, this will imply that the final error is at most $8\gamma$.
\begin{proof}[Proof of Theorem~\ref{thm:boosting-SDP}]
Let $\theta$ be the fraction of entries where $\ell_{\text{init}}$ and $\wt{\ell}$ disagree.  Let $T$ be the subset of $S$ where $\ell_{\text{init}}$ agrees with $\wt{\ell}$.  By the assumptions in the theorem, we can apply Lemma~\ref{lem:exists-solution} on the subset of entries indexed by $T \times T$ to get that the optimal solution to the boosting SDP has objective value at most
\[
\rho \leq \gamma + \theta \,.
\]
Now we argue about the structure of the optimal solution using Lemma~\ref{lem:progress}.  In particular, any solution with $\rho \leq \theta + \gamma$ must place weight $w_i \geq 1 - 1/\sqrt{K}$ on all but at most 
\[
10(\theta + \gamma)n/\sqrt{K} \leq 0.1(\theta + \gamma)n
\]
elements of $S\backslash T$.  Note that this means the total weight $w_i$ on indices where $\ell_{\text{init}}$ and $\wt{\ell}$ agree is at most 
\[
(\theta + \gamma)n - \left(1 - \frac{1}{\sqrt{K}} \right)( \theta - \gamma - 0.1(\theta + \gamma))n \leq (0.11\theta + 2.2 \gamma) n   
\]
where the second term above comes from the fact that there are $\theta n$ entries where $\ell_{\text{init}}$ and $\wt{\ell}$ disagree, at most $\gamma n$ of these are outside of $S$ and at most $0.1(\theta + \gamma)n$ of these that are inside $S$ have $w_i < 1 - 1/\sqrt{K}$.  Thus, after the flipping, the new labelling $\ell$ has error at most
\[
\theta n - ( \theta - \gamma - 0.1(\theta + \gamma))n + \frac{1}{1 - 1/\sqrt{K}} \cdot (0.11\theta + 2.2 \gamma) n \leq (0.3\theta +4\gamma)n \,.
\]
Now we can apply the above argument for each iteration of the for loop in the {\sc Boosting using SDP} algorithm.  Note that the error is cut by a factor of $0.8$ whenever  $\theta \geq 8 \gamma$ and also once we have error $\theta \leq 8\gamma$, the error will never go above $8\gamma$.  Since we run $10 \log n$ iterations and any nonzero error is at least $1/n$, we get immediately get the desired bound and are done.
\end{proof}

\subsection{Complete Analysis for $2$ Communities}\label{sec:2community-full}

Combining Theorem~\ref{thm:boosting-SDP} with the tail bounds in Section~\ref{sec:tail-bounds} (in particular Corollary~\ref{coro:resolvable2} and Corollary!\ref{coro:spectral-remove}) and our initialization algorithm in Section~\ref{sec:initialization}, we can prove our main theorem for robust community detection with $2$ communities.  The steps of our algorithm are summarized below.

\begin{algorithm}[H]
\caption{{\sc Full Robust Community Detection ($k = 2$)} }
\begin{algorithmic}
\State \textbf{Input:} Adjacency matrix $A \in \R^{n \times n}$, parameters $a,b, \eps, \alpha$
\State Set $\wh{A} = A - D(a/n,b/n) J$
\State Set $C = (\sqrt{a} - \sqrt{b})^2$
\State Run {\sc Compute Initial Labelling} to compute labelling $\ell_{\text{init}} \in \{-1,1 \}^n$
\If{$\eps \geq 1/\sqrt{C}$}
\State \textbf{Output:} $\ell_{\text{init}}$
\Else 
\State Run {Boosting Using SDP} on $\wh{A}, \ell_{\text{init}}$ and parameters
\begin{align*}
d & \leftarrow  \sqrt{a + b} \\
\zeta & \leftarrow \frac{4 \cdot 10^5}{\alpha^3} \cdot \frac{\chi}{\sqrt{C}} \\  
K & \leftarrow  \frac{10^6\chi}{\alpha^3}
\end{align*}
where $\chi$ is a (sufficiently large) universal constant
\State \textbf{Output:} final labelling $\ell$
\EndIf
\end{algorithmic}
\label{alg:full-k=2}
\end{algorithm}

\begin{proof}[Proof of Theorem~\ref{thm:main-SBM1}]
By Lemma~\ref{lem:rough-clustering}, with probability $ 1 - 2/n^2$, the accuracy of the initial clustering is at least 
\[
1 -\frac{10^4  k}{\alpha^3}\left( \frac{\chi}{\sqrt{C}} + \eps \right) \,.
\]
Thus, in the case where $\eps \geq 1/\sqrt{C}$, we are immediately done.  Otherwise, the above accuracy is at least $1 - 0.1\zeta$.  Now it suffices to verify the remaining conditions of Theorem~\ref{thm:boosting-SDP}.  Let $\wt{\ell}$ denote the true labelling and let $\wt{L} = \wt{\ell}\wt{\ell}^T$.  Let $\kappa$ be a parameter that will allow us to balance the accuracy with the failure probability from the generative model.  Set 
\[
\gamma = e^{-C/2 + 3\kappa + (10K)^3\sqrt{a + b} \log R(a/n,b/n)} + \eps + \frac{10^4}{n} \max\left(0, \frac{ \kappa}{\sqrt{C}} - 1 \right) \,.
\]
In the generation of the $\eps$-corrupted SBM, let $A_0$ be the initial, pure SBM adjacency matrix (before semi-random noise and corruptions).  By Corollary~\ref{coro:spectral-remove}, with probability at least $1 - 2/ n^2$, we can find a subset $S_0$ of  $(1 - e^{-2C}) n$ nodes such that 
\[
\norm{\left(A_0 - \frac{a + b}{2n}J - \frac{a - b}{2n}\wt{L} \right)_{S_0 \times S_0}}_{\op} \leq \chi \sqrt{a + b} \,.
\]
Now let $S$ be the subset of $S_0$ of uncorrupted nodes (after the adversary makes the $\eps$-corruption).  Note that clearly $|S| \geq (1 - \gamma )n $ by the earlier definition of $\gamma$.  By definition, we can write
\[
\wh{A} = A - D(a/n,b/n)J = F + \left( A_0 - \frac{a + b}{2n}J - \frac{a - b}{2n} \wt{L} \right) +  \left(\frac{a + b}{2n}J + \frac{a - b}{2n} \wt{L} - D(a/n,b/n)J \right)
\]
where $F$ is a matrix such that $(F \odot \wt{L})_{S \times S}$ is entry-wise nonnegative. We can now set
\begin{align*}
Y &= \left( A_0 - \frac{a + b}{2n}J - \frac{a - b}{2n} \wt{L} \right)_{S \times S} \\
Z &= \left(\frac{a + b}{2n}J + \frac{a - b}{2n} \wt{L} - D(a/n,b/n)J \right)_{S \times S} \,.
\end{align*}
Note that by Claim~\ref{claim:log-bound}, all entries of  $Z$ are in the interval $\left[ -\frac{a - b}{n}, \frac{a -b}{n}\right]$ and thus, the above completes the verification of the second and third properties that we need to apply Theorem~\ref{thm:boosting-SDP}.  

Now, it remains to verify the final property in order to apply Theorem~\ref{thm:boosting-SDP}.  To do this, we rely on Corollary~\ref{coro:resolvable2}.  Recall that 
\[
\log R(a/n,b/n) \geq \frac{2(\sqrt{a} - \sqrt{b})}{\sqrt{a + b}} = \frac{2\sqrt{C}}{\sqrt{a + b}} \,.
\]
Also note that 
\[
(\wh{A} - F)_{S \times S} = (A_0 - D(a/n,b/n)J)_{S \times S} \,.
\]
Using the above and setting 
\[
\theta = e^{-C/2 + 3\kappa + (10K)^3\sqrt{a + b} \log R(a/n,b/n)}
\]
in Corollary~\ref{coro:resolvable2} immediately implies that for any subset $T \subset S$ with $|T| \geq (1 - \zeta)n$, the matrix $((\wh{A} - F) \odot \wt{L})_{T \times T} $ is resolvable with parameters 
\[
\left( 100K^3\sqrt{ a + b},  1.1 \left(\left( \theta + \frac{n - |T|}{n} \right)\sqrt{a  +b} + \frac{\max(0, 10^4(\kappa - \sqrt{C}))}{n \log R(a/n,b/n)}\right) \right)
\]
with probability at least $1 - e^{-10\kappa} - 2/n^2$. Substituting in the definitions of $d$ and $\gamma$ completes the verification of the last property that we need in order to apply Theorem~\ref{thm:boosting-SDP}.

We conclude that the accuracy of the final labelling $\ell$ is at least $1 - 8\gamma$ with probability at least $1 - e^{-10\kappa} - 4/n^2$.  Finally, to complete the proof and bound the expected accuracy, it suffices to substitute in the expression for $\gamma$ and integrate over the failure probability (which is controlled by $\kappa$) and we get that the expected accuracy is at least 
\begin{equation}\label{eq:error-bound1}
1 -  8\eps - e^{-C/2  + (100K)^3\sqrt{a + b} \log R(a/n,b/n)}  - \frac{e^{-\sqrt{C}}}{n} \,.
\end{equation}
Finally, note that if $a/b \geq C$, then we can imagine replacing $b$ with $a/C$ and pretending that there is more semi-random noise.  Note that
\[
(\sqrt{a} - \sqrt{a/C})^2 \geq a \left(1 - \frac{2}{\sqrt{C}} \right) \geq C - 2\sqrt{C} \,.
\]
Thus, replacing $C \leftarrow C - 2\sqrt{C}$ if necessary, we can ensure $a/b \leq C$.  If this happens then
\[
\log R(a/n,b/n) \leq 4 \log \sqrt{a/b} \leq 2 \log C \cdot \frac{\sqrt{a} - \sqrt{b}}{\sqrt{a + b}} = 2 \log C \cdot \frac{\sqrt{C}}{\sqrt{a + b}}
\]
where the middle inequality above holds because $\log \sqrt{a/b} \leq \sqrt{a/b} - 1 = (\sqrt{a} - \sqrt{b})/\sqrt{b}$ which immediately implies the above if $a  <3b$ and if $a \geq 3b$ then 
\[
\log \sqrt{a/b} \leq \frac{1}{2}\log C \leq 2\log C \cdot \frac{\sqrt{a} - \sqrt{b}}{\sqrt{a + b}}
\]
where we used the assumption that $a/b \leq C$.  Substituting our bound on $\log R(a/n,b/n)$ back into \eqref{eq:error-bound1} and noting that $K = O(\alpha^{-3})$, we conclude that the expected accuracy is at least
\[
1 - 8\eps - \frac{e^{-\sqrt{\log n}}}{n} - e^{-C/2  + O(\alpha^{-9} \sqrt{C} \log C)}
\]
and we are done.
\end{proof}

%% file: boosting-SDP-kcommunity.tex
\section{The Robust Boosting Procedure for $k$ Communities}\label{sec:boost-kcommunity}

Now we consider the case of community detection with $k > 2$ communities.  In this case, we use a generalization of our boosting SDP for two communities where we consider all pairs of $\binom{k}{2}$ communities and essentially enforce the constraints of the two community SDP on each pair of communities (based on our rough estimate of the communities).  We first introduce some notation. 
\begin{definition}\label{def:associated-matrices}
Let $\mcl{P} = \{S_1, \dots , S_k \}$ be a partition of $[n]$ into $k$ parts.  For distinct $j_1,j_2 \in [k]$, we define the associated vectors $\ell_{j_1j_2}$ of $\mcl{P}$ to be vectors in $\{-1,0,1 \}^n$ that have $1$ in entries indexed by elements of $S_{j_1}$, $-1$ in entries indexed by elements of $S_{j_2}$ and $0$ everywhere else.  We define the associated matrices as $L(j_1,j_2) = \ell_{j_1j_2} \ell_{j_1j_2}^T$.
\end{definition}

Now we formally state the boosting SDP.  
\begin{definition}[Boosting SDP For $k$ Communities]
Let $\wh{A}$ be an $n \times n$ matrix we are given as input.  Let $\mcl{P} = \{S_1 ,\dots , S_k \}$ be a partition of $[n]$ into $k$ parts that we are given.  We let $L(j_1,j_2)$ for distinct $j_1,j_2 \in [k]$ be the associated matrices of $\mcl{P}$.  Let $\alpha , \zeta, d, K$ be some parameters that we can set.  The boosting SDP for $\wh{A},S_1, \dots , S_k$ and parameters $\alpha , \zeta , d,  K$ with $0 < \zeta < 1, K \geq 10^4$ is defined as follows: we have a variable $0 \leq \rho \leq \zeta$, weights $w_1, \dots , w_n \in \R$ and matrices $W \in \R^{n \times n}, N \in \mcl{R}_n(\rho)$ such that
\begin{align*}
&0 \leq w_1, \dots , w_n \leq 1 \\
&w_1 + \dots + w_n \leq \rho  n\\
&W_{ij} \geq 0 \quad \forall i,j \\
&W = J - \begin{bmatrix}w_1 \one & \dots & w_n \one \end{bmatrix} - \begin{bmatrix}w_1 \one^T \\ \vdots \\ w_n \one^T \end{bmatrix} + N
\end{align*}
and finally we have the constraint that for all distinct $j_1, j_2 \in [k]$ and all $\rho'$ with $k \rho/(\alpha K) \leq \rho' \leq \zeta $ and all 
\[
(M, x_1, \dots , x_n) \in \mcl{S}_n\left(\frac{\rho' (|S_{j_1}| + |S_{j_2}|)}{n} , \frac{K \rho' (|S_{j_1}| + |S_{j_2}|)}{n} \right)
\]
that
\[
 \left\la \wh{A}  \odot L(j_1,j_2) \odot W, M  \right\ra \geq  10d K^2\left( K\left(\sum_{i \in S_{j_1} \cup S_{j_2}} x_i ( 1 - w_i) \right) -   \rho' (|S_{j_1}| + |S_{j_2}|) \right)  \,.
\]
The objective is to find a feasible solution that minimizes $\rho$.
\end{definition}

The algorithm and analysis are very similar to the $k = 2$ case.  The main intuition is that for any solution $W$ to the $k$-community boosting SDP, for any distinct $j_1, j_2 \in [k]$, the restriction of $W$ to vertices in $S_{j_1} \cup S_{j_2}$ is a solution to the $2$-community boosting SDP on $S_{j_1} \cup S_{j_2}$.  Thus, we can still use Lemmas~\ref{lem:exists-solution} and \ref{lem:progress} to argue about properties of solutions to the boosting SDP.   

\begin{algorithm}[H]
\caption{{\sc Boosting Using SDP $k$-community} }
\begin{algorithmic}
\State \textbf{Input:} Matrix $\wh{A}$, parameters $\alpha , \zeta , d , K$
\State \textbf{Input:} Rough partition of $[n]$ into $k$ parts given by $\mcl{P}_{\text{init}}$ 
\State Set $\mcl{P} \leftarrow \mcl{P}_{\text{init}}$
\For{$t = 1,2, \dots ,  10 k \log n$}
\State Run Boosting SDP with $\mcl{P} , \wh{A}, \alpha , \zeta,d , K$
\For{ $i \in [n]$ such that $w_i \geq 1 - 1/\sqrt{K}$,}
\State Flip the label to one of the other $k-1$ possibilities uniformly at random 
\EndFor
\State Set $\mcl{P}$ to the new labelling (after flips)
\EndFor
\State \textbf{Output:} $\mcl{P}$
\end{algorithmic}
\label{alg:round-boostingSDP2}
\end{algorithm}

The analysis relies on two key lemmas which parallel those in the previous section.  Lemma~\ref{lem:exists-solution2} parallels Lemma~\ref{lem:exists-solution} and shows the existence of a solution with good objective value.  Lemma~\ref{lem:progress2} parallels Lemma~\ref{lem:progress} and gives structural properties that must hold for any solution with good objective value.

\begin{lemma}\label{lem:exists-solution2}
Consider solving the $k$-community boosting SDP with input partition $\mcl{P} = \{S_1, \dots , S_k \}$.  Assume that $|S_j| \geq \alpha n /k$ for all $j \in [k]$ and there is an $n \times n$ matrix $F$ and a subset $S \subset [n]$ with $|S| \geq (1 - \theta)n$ and $\theta \leq \alpha \zeta / k$ such that the following properties hold: 
\begin{itemize}
    \item For all distinct $j_1, j_2 \in [k]$, $(F \odot L(j_1,j_2))_{S \times S}$ is entrywise nonnegative
    \item $(\wh{A} - F)_{S \times S}  = Y + Z$ where $\norm{Y}_{\op} \leq K d$ and $|Z_{ij}| \leq K d/(\zeta n)$ for all $i,j \in [n]$
    \item For all distinct $j_1, j_2 \in [k]$, the matrix
    \[
    \left( (\wh{A} - F) \odot L(j_1,j_2) \right)_{(S \cap (S_{j_1} \cup S_{j_2})) \times (S \cap (S_{j_1} \cup S_{j_2}))}
    \]
    is $(10 d K^3  , \theta d K )$-resolvable 
\end{itemize}
Then there is a feasible solution to the SDP with $\rho  = \theta$ where $W$ is the $n \times n$ matrix that is $1$ on the square corresponding to $S \times S$ and $0$ everywhere else.  
\end{lemma}
\begin{proof}
As in Lemma~\ref{lem:exists-solution}, we can let $w_i = 1$ for all elements $i \in [n]\backslash S$ and $0$ otherwise.  We then set $N$ to be the indicator of the square indexed by $[n]\backslash S \times [n]\backslash S$.  Now we show how to apply Lemma~\ref{lem:exists-solution} to check feasibility.  It suffices to consider fixed indices $j_1, j_2$.  Let $R = S_{j_1} \cup S_{j_2}$ and let $r = |R|$.  Note that by the assumptions in the lemma, we must have
\[
|S \cap R| \geq (1 - k\theta/\alpha) r  \geq (1 - \zeta) r \,.
\]
Now, to check the constraint of the SDP, it suffices to restrict $\wh{A}, W, M$ and $x_1, \dots , x_n$ to those indexed by elements of $R$.  For any $(M, x_1, \dots , x_n)$ for which we need to check the constraint, we have
\[
(M_{R \times R}, \{ x_i \}_{i \in R} ) \in \mcl{S}_r( \rho', K \rho' ) 
\]
and $ k \theta /(\alpha K) \leq \rho' \leq \zeta$.  Thus, we can apply Lemma~\ref{lem:exists-solution} (with $\theta \leftarrow k\theta / \alpha$) with the matrices $\wh{A}_{R \times R}, F_{R \times R}, L(j_1, j_2)_{R \times R}$ and subset $S \cap R$ to deduce that the desired constraint is satisfied.
\end{proof}

%\begin{definition}
%For two partitions $\mcl{P} = \{S_1, \dots , S_k \}$ and $\mcl{P}' = \{S_1' , \dots, S_k' \}$ of $[n]$ into $k$ sets, we say they are $\theta$-close if for all $i \in [k]$, we have
%\[
%|S_i \backslash S_i'| + |S_i' \backslash S_i| \leq \theta \min(|S_i|, |S_i'|) \,.
%\]
%\end{definition}

\begin{lemma}\label{lem:progress2}
Consider solving the $k$-community boosting SDP and assume that the input partition has all communities of size at least $\alpha n /k$.  Let $\wt{\mcl{P}} = \{\wt{S_1}, \dots , \wt{S_k} \}$ be a partition of $[n]$ into $k$ parts that disagrees with $\mcl{P}$ on at most $\theta n$ elements where $\theta \leq 0.5\alpha \zeta/k$.  For distinct $j_1, j_2 \in [k]$, let $\wt{\ell_{j_1j_2}}$ be the associated vectors (recall Definition~\ref{def:associated-matrices}) and $\wt{L(j_1, j_2)}$ be the associated matrices of $\wt{\mcl{P}}$.   Assume that for some $\gamma \leq 0.5\alpha \zeta/k$, there is a subset $S \subset [n]$ with $|S| \geq (1 - \gamma )n$ and an $n \times n$ matrix $F$ such that the following properties hold:
\begin{itemize}
    \item For all distinct $j_1, j_2 \in [k]$, $(F \odot \wt{L(j_1, j_2)})_{S \times S}$ is entrywise nonnegative
     \item $(\wh{A} - F)_{S \times S}  = Y + Z$ where $\norm{Y}_{\op} \leq Kd$ and $|Z_{ij}| \leq K d/(\zeta n )$ for all $i,j \in [n]$
    \item For all distinct $j_1, j_2 \in [k]$, the matrix
    \[
    \left( (\wh{A} - F) \odot \wt{L(j_1, j_2)} \right)_{S(j_1, j_2) \times S(j_1, j_2)}
    \]
    is $(10 dK^3, (\theta + \gamma) d K)$-resolvable where $S(j_1, j_2) = S \cap (S_{j_1} \cup S_{j_2}) \cap (\wt{S_{j_1}} \cup \wt{S_{j_2}})$
\end{itemize}
Then for any feasible solution to the boosting SDP with $\rho \leq (\theta + \gamma)$, there are at most $20(\theta + \gamma)k^3 n/(\alpha \sqrt{K})$ elements $i \in S$ such that $w_i \leq 1 - 1/\sqrt{K}$ and where $\mcl{P}$ and $\wt{\mcl{P}}$ disagree.
\end{lemma}
\begin{proof}
We will consider each pair of distinct $j_1, j_2 \in [k]$ and apply Lemma~\ref{lem:progress} to each.  First, fix $j_1, j_2$.  Let $R = S_{j_1} \cup S_{j_2}$.  Note that by definition, $W_{R \times R}$ is a feasible solution to the $2$-community boosting SDP for the matrix $\wh{A}_{R \times R}$ and the labelling $\ell_{j_1j_2}$  with objective value at most $k \rho/\alpha$.  Thus, we can apply Lemma~\ref{lem:progress} on the solution $W_{R \times R}$ restricted to the set of entries in $R \times R$.  Note that when we apply Lemma~\ref{lem:progress} we are setting 
\begin{itemize}
    \item $S \leftarrow S(j_1, j_2)$ and $\wt{\ell} \leftarrow \wt{\ell_{j_1j_2}}$ (technically $\wt{\ell_{j_1j_2}}$ is not a full labelling but it is a full labelling on $S(j_1,j_2)$ so we can still apply the lemma)
    \item $\theta \leftarrow k \theta / \alpha$, $\gamma \leftarrow k (\theta + \gamma) / \alpha$
\end{itemize}
where we use that
\[
|S(j_1, j_2)| \geq |R| - (\theta + \gamma )n \geq (1 - k(\theta + \gamma)/\alpha) |R| \,.
\]
We conclude that among all vertices $i \in S \cap (S_{j_1} \cap \wt{S_{j_2}})$ or $i \in S \cap (\wt{S_{j_1}} \cap S_{j_2})$, at most $20(\theta + \gamma)k  n/(\alpha \sqrt{K})$ of them have weight $w_i \leq 1 - 1/\sqrt{K}$.  Now we can simply sum this over all choices of $j_1, j_2$ to complete the proof of the desired statement.
\end{proof}

Now we can prove the analogue of Theorem~\ref{thm:boosting-SDP}, giving explicit conditions on the input under which the boosting procedure succeeds.  The analysis will need to a bit more precise than the analysis in Theorem~\ref{thm:boosting-SDP} since flipping labels does not guarantee that we correctly label the vertex but only increases the probability that we correctly label them.  Fortunately, by choosing $K$ sufficiently large, we can ensure that almost all of the labels that we flip are originally wrong so this still makes progress in expectation.  We then argue that if we run a sufficient number of iterations, we make progress with high probability.
\begin{theorem}\label{thm:boosting-SDP2}
Assume that we run {\sc Boosting Using SDP $k$-community} and that the input has all communities of size at least $\alpha n /k$, the parameters satisfy $K \geq (10k/\alpha)^{10}, \zeta \geq 1/\sqrt{n}$ and there is a subset $S \subset [n]$ with $|S| \geq (1 - \gamma)n$ where $\gamma \leq 0.1 \alpha \zeta /k$, an $n \times n$ matrix $F$ and a partition of $[n]$ given by $\wt{\mcl{P}} = \{\wt{S_1}, \dots , \wt{S_k} \}$ with associated matrices $\wt{L(j_1, j_2)}$ such that 
\begin{enumerate}
    \item $\mcl{P}_{\text{init}}$ agrees with $\wt{\mcl{P}}$ on at least $(1 - 0.1\alpha \zeta/k)n $ entries
    \item For any distinct $j_1, j_2 \in [k]$, $(F \odot \wt{L(j_1, j_2)})_{S \times S}$ is entrywise nonnegative
    \item $(\wh{A} - F)_{S \times S} = Y + Z$ where $\norm{Y}_{\op} \leq Kd$ and $|Z_{ij}| \leq Kd/(\zeta n )$ for all $i,j \in [n]$
    \item For any distinct $j_1, j_2 \in [k]$ and all subsets $U \subset S \cap (\wt{S_{j_1}} \cup \wt{S_{j_2}})$ with $|U| \geq (1 - \zeta)(|\wt{S_{j_1}}| + |\wt{S_{j_2}}|)$ we have that $\left((\wh{A} - F)\odot \wt{L(j_1, j_2)} \right)_{ U \times U}$ is resolvable with parameters
    \[
    \left( 10dK^3, 0.5dK \left(\gamma + \frac{|\wt{S_{j_1}}| + |\wt{S_{j_2}}| - |U|}{n} \right) \right)
    \]
\end{enumerate}
Then the expected accuracy of the output of {\sc Boosting Using SDP $k$-community} (with respect to $\wt{\mcl{P}}$) is at least $(1 - 8k\gamma - 1/n^2)$.
\end{theorem}
\begin{proof}
Let $T_0 \subset [n]$ be the set of elements where $\mcl{P}_{\text{init}}$ and $\wt{\mcl{P}}$ agree and let $|[n] \backslash T_0| = \theta n$.  Let $T = S \cap T_0$.  Now, we apply Lemma~\ref{lem:exists-solution2} with the subset $T$.  To see that this application is valid, note that $\theta \leq 0.1\alpha \zeta n/k$ by assumption.  The first two conditions that we need to check for the lemma are trivial because by construction,
\[
L(j_1, j_2)_{T \times T} = \wt{L(j_1, j_2)}_{T \times T} \,.
\]
To check the last condition about resolvability, note that for any distinct $j_1, j_2 \in [k]$, the definition of $T$ implies
\begin{align*}
T \cap (S_{j_1} \cup S_{j_2}) &\subset S \cap (\wt{S_{j_1}} \cup \wt{S_{j_2}}) \\
|T \cap (S_{j_1} \cup S_{j_2})| &\geq |\wt{S_{j_1}}| + |\wt{S_{j_2}}| - (\theta + \gamma)n  \geq (1 - \zeta)(|\wt{S_{j_1}}| + |\wt{S_{j_2}}| )
\end{align*}
so we can simply set $U = T \cap (S_{j_1} \cup S_{j_2})$ in the fourth assumption in the theorem statement.  We conclude that the optimal solution to the boosting SDP has objective value at most 
\[
\rho \leq \theta + \gamma \,.
\]

Next, we apply Lemma~\ref{lem:progress2} to argue about the structure of the optimal solution.  When we are applying Lemma~\ref{lem:progress2} we are setting $\wt{\mcl{P}} \leftarrow \wt{\mcl{P}}$, $S \leftarrow S$.  The first two properties that we need to verify to apply the lemma follow trivially from our assumptions.  To verify the last one about resolvability, note that 
\[
|S \cap (S_{j_1} \cup S_{j_2}) \cap (\wt{S_{j_1}} \cup \wt{S_{j_2}})| \geq |\wt{S_{j_1}}| + |\wt{S_{j_2}}| - (\theta + \gamma)n \geq (1 - \zeta)(|\wt{S_{j_1}}| + |\wt{S_{j_2}}| ) 
\]
so we can plug in $U = S \cap (S_{j_1} \cup S_{j_2}) \cap (\wt{S_{j_1}} \cup \wt{S_{j_2}})$ into the fourth assumption in the theorem statement.  Thus, we can apply Lemma~\ref{lem:progress2} and we conclude that any solution with $\rho \leq \theta  +\gamma $ places weight $w_i \geq 1 - 1/\sqrt{K}$ on all but at most 
\[
\frac{ (\theta + \gamma) n}{100k^2}
\]
elements of $S\backslash T$.  Thus, the total weight $w_i$ on indices where $\mcl{P}_{\text{init}}$ and $\wt{\mcl{P}}$ agree is at most
\[
(\theta + \gamma)n - \left(1 - \frac{1}{\sqrt{K}} \right)\left( \theta - \gamma - \frac{ (\theta + \gamma) }{100k^2}\right)n  \leq \left(\frac{\theta}{10k^2} + 2.2\gamma \right)n \,.
\]
Now, after flipping, the expected error (over the random choices of the flips) of the new labelling is at most
\[
\theta  - \frac{1}{k} \left(\theta - \gamma - \frac{ (\theta + \gamma) }{100k^2}  \right) + 1.1 \cdot \left(\frac{\theta}{10k^2} + 2.2\gamma \right) \leq \left(1 - \frac{1}{2k} \right) \theta + 3\gamma  \,.
\]
Of course, we can apply the above argument for each iteration of the for loop in {\sc Boosting Using SDP $k$-community}.  Since we run $10 k \log n$ iterations, the above recurrence implies that the expected error at the end is at most $8k\gamma + 1/n^2$.  There is one additional technicality that we must deal with, which is to ensure that the error never goes above the threshold of $0.5 \alpha \zeta / k$ since if this happens then Lemma~\ref{lem:exists-solution2} and Lemma~\ref{lem:progress2} can no longer be applied.  Since we assumed that $\zeta \geq 1/\sqrt{n}$, a Chernoff bound ensures that the probability of this bad event happening is at most $e^{-\Omega(\alpha \sqrt{n}/ k)}$ which is negligible.  
\end{proof}

\subsection{Completing the Analysis}
Finally, we can complete the proof of Theorem~\ref{thm:main-SBM2}.  The proof is very similar to that of Theorem~\ref{thm:main-SBM1}.  We will use the initialization algorithm in Section~\ref{sec:initialization} and then use results in Section~\ref{sec:tail-bounds} to verify that the conditions necessary to apply Theorem~\ref{thm:boosting-SDP2} for the boosting step hold for a matrix generated from an $\eps$-corrupted SBM.   The algorithm is also exactly the same as the algorithm for the $2$-community case except we use {\sc Boosting Using SDP $k$-Community} for the boosting step.

\begin{algorithm}[H]
\caption{{\sc Full Robust Community Detection ($k > 2$)} }
\begin{algorithmic}
\State \textbf{Input:} Adjacency matrix $A \in \R^{n \times n}$, parameters $a,b, \eps, k, \alpha $
\State Set $\wh{A} = A - D(a/n,b/n) J$
\State Set $C = (\sqrt{a} - \sqrt{b})^2$
\State Run {\sc Compute Initial Labelling} to compute partition $\mcl{P}_{\text{init}} = \{S_1, \dots , S_k \}$ of $[n]$
\If{$\eps \geq 1/\sqrt{C}$}
\State \textbf{Output:} $\mcl{P}_{\text{init}}$
\Else 
\State Run {Boosting Using SDP $k$-Community} on $\wh{A}, \mcl{P}_{\text{init}}$ and parameters
\begin{align*}
\alpha & \leftarrow \alpha  \\
d & \leftarrow  \sqrt{a + b} \\
\zeta & \leftarrow\frac{2 \cdot 10^5 k^2 }{\alpha^4} \cdot \frac{\chi}{\sqrt{C}} \\  
K & \leftarrow \left(\frac{10 k}{\alpha} \right)^{10} \cdot \chi
\end{align*}
where $\chi$ is a (sufficiently large) universal constant
\State \textbf{Output:} final labelling $\mcl{P}$
\EndIf
\end{algorithmic}
\label{alg:full-general}
\end{algorithm}

\begin{proof}[Proof of Theorem~\ref{thm:main-SBM2}]
By Lemma~\ref{lem:rough-clustering}, with probability $ 1 - 2/n^2$, the accuracy of the initial clustering is at least 
\[
1 - \frac{k \cdot 10^4 }{\alpha^3} \cdot \left(\frac{\chi }{\sqrt{C}} + \eps \right) \,.
\]
Thus, in the case where $\eps \geq 1/\sqrt{C}$, we are immediately done.  Otherwise, the above accuracy is at least $1 - 0.1\alpha \zeta/k$.  Also, note that $\zeta \geq 1/\sqrt{n}$ since it suffices to consider when $C = o(n)$.  Now we will verify the remaining conditions to apply Theorem~\ref{thm:boosting-SDP2}.  Let $\wt{\mcl{P}}$ denote the true labelling and let $\wt{L(j_1, j_2)}$ be its associated matrices.  Let $\wt{L}$ be the matrix that has $1$ in entries indexed by two vertices in the same (true) community and $-1$ in other entries.  Note that the matrices $\wt{L(j_1,j_2)}$ are all obtained by restricting to a submatrix of $\wt{L}$ indexed by $(\wt{S_{j_1}} \cup \wt{S_{j_2}}) \times (\wt{S_{j_1}} \cup \wt{S_{j_2}})$ and zeroing out everything else. 
\\\\
Let $\kappa$ be a parameter that will allow us to balance the accuracy with the failure probability from the generative model.  Set 
\[
\gamma = e^{-\alpha C/k + 3\kappa + (10K)^5\sqrt{a + b} \log R(a/n,b/n)} + \eps + \frac{10^4}{n} \max\left(0, \frac{k \kappa}{\alpha \sqrt{C}} - 1 \right) \,.
\]
In the generation of the $\eps$-corrupted SBM, let $A_0$ be the initial, pure SBM adjacency matrix (before semi-random noise and corruptions).  By Corollary~\ref{coro:spectral-remove}, with probability at least $1 - 2/ n^2$, we can find a subset $S_0$ of  $(1 - e^{-2C}) n$ nodes such that 
\[
\norm{\left(A_0 - \frac{a + b}{2n}J - \frac{a - b}{2n}\wt{L} \right)_{S_0 \times S_0}}_{\op} \leq \chi\sqrt{a + b} \,.
\]
Now let $S$ be the subset of $S_0$ of uncorrupted nodes (after the adversary makes the $\eps$-corruption).  Note that clearly $|S| \geq (1 - \gamma )n $ and $\gamma \leq 0.1\alpha \zeta/k$ by the earlier definition of $\gamma$.  We can write
\[
\wh{A} = A - D(a/n,b/n)J = F + \left( A_0 - \frac{a + b}{2n}J - \frac{a - b}{2n}\wt{L} \right) +  \left(\frac{a + b}{2n}J + \frac{a - b}{2n}\wt{L} - D(a/n,b/n)J \right)
\]
where $F$ is a matrix such that $(F \odot \wt{L})_{S \times S}$ is entry-wise nonnegative.  We can now set
\begin{align*}
Y &= \left( A_0 - \frac{a + b}{2n}J - \frac{a - b}{2n}\wt{L} \right)_{S \times S} \\
Z &= \left(\frac{a + b}{2n}J + \frac{a - b}{2n}\wt{L} - D(a/n,b/n)J \right)_{S \times S} \,.
\end{align*}
By Claim~\ref{claim:log-bound}, all entries of  $Z$ are in the interval $\left[ -\frac{a - b}{n}, \frac{a -b}{n}\right]$ and plugging in the parameter settings, we have $Kd/(\zeta n) \geq (a - b)/n$.  Thus, the above completes the verification of the second and third properties that we need to apply Theorem~\ref{thm:boosting-SDP2}.  

Now, it remains to verify the final property.  We will apply Corollary~\ref{coro:resolvable2} on each pair of communities $\wt{S_{j_!}}, \wt{S_{j_2}}$.  Fix $j_1, j_2 \in [k]$.  Let $r = |\wt{S_{j_1}}| + |\wt{S_{j_2}}|$.  Recall that 
\[
\log R(a/n,b/n) \geq \frac{2(\sqrt{a} - \sqrt{b})}{\sqrt{a + b}} = \frac{2\sqrt{C}}{\sqrt{a + b}} \,.
\]
Also note that 
\[
(\wh{A} - F)_{S \times S} = (A_0 - D(a/n,b/n)J)_{S \times S} \,.
\]
Setting
\[
\theta = e^{- \frac{rC}{2n} + 3\kappa + (10K)^4 \sqrt{(a + b)r/n}\log R(a/n,b/n)  }
\]
and using Corollary~\ref{coro:resolvable2} immediately implies that for any subset $U \subset S \cap (\wt{S_{j_1}} \cup \wt{S_{j_2}})$ with $|U| \geq (1 - \zeta)r$, the matrix 
\[
((\wh{A} - F) \odot \wt{L})_{(S \cap (\wt{S_i} \cup \wt{S_j})) \times (S \cap (\wt{S_i} \cup \wt{S_j}))} 
\]
is resolvable with parameters
\[
\left( 5 \cdot 10^3 K^4 \sqrt{(a + b)r/n}, 1.1 \left(  \left(\theta + \frac{r - |U|}{r}  \right) \right)\sqrt{(a + b)r/n} + \frac{\max(0, 10^4(\kappa - \sqrt{rC/n})}{r \log R(a/n,b/n)} \right)
\]
with probability at least $1 - e^{-10\kappa} - 2/r^2$ (note that we are applying Corollary~\ref{coro:resolvable2} on an $r \times r$ matrix so $a,b, C$ get scaled down by a factor of $r/n$ accordingly).  We must have $r \geq 2\alpha n /k$.  Thus, using this and the definition of $\gamma, \theta$ and the settings of parameters $K,d, \zeta$, we get that the fourth condition in Theorem~\ref{thm:boosting-SDP2} is indeed satisfied.
\\\\
Overall, we can union bound the above over all pairs $j_1, j_2$ to get a probability of 
\[
1 - k^2 e^{-10\kappa} - k^4/(\alpha^2 n^2)
\]
that the condition holds for all pairs $j_1,j_2$ simultaneously.  Finally, assuming that this holds, we can apply Theorem~\ref{thm:boosting-SDP2} to get that the expected accuracy of the final labelling $\mcl{P}$ is at least $1 - 8k\gamma - 1/n^2$.  Finally, to complete the proof and bound the expected accuracy, it suffices to substitute in the expression for $\gamma$ and integrate over the failure probability (which is controlled by $\kappa$).  Using essentially the same computations as in the proof of Theorem~\ref{thm:main-SBM1}, we get that the expected accuracy of the algorithm is at least 
\[
1 - O(\eps k/\alpha^3) -  e^{-\alpha C/k  + \poly(k/\alpha) \sqrt{C} \log C} - \frac{e^{-\sqrt{\log n}}}{n} 
\]
and this completes the proof.
\end{proof}

%% file: appendix.tex
\appendix

\section{Deferred Proofs from Section~\ref{sec:prelims}}\label{appendix:prelims}

\begin{proof}[Proof of Claim~\ref{claim:log-bound}]
Without loss of generality $p > q$.  We have
\begin{equation}\label{eq:log-formula}
D(p,q) = \frac{1}{1 + \frac{\log p - \log q}{\log(1 - q) - \log(1 - p)}} \,.
\end{equation}
Since $\log x$ is concave, we have 
\[
\frac{p - q}{p} \leq \log p - \log q \leq \frac{p - q}{q}
\]
and
\[
\frac{p - q}{ 1 - q}\leq \log(1 - q) - \log(1 - p) \leq \frac{p - q}{1 - p} \,.
\]
Plugging in the above two inequalities into~\eqref{eq:log-formula} gives
\[
q \leq D(p,q) \leq p
\]
as desired.
\end{proof}

\begin{proof}[Proof of Claim~\ref{claim:binomials-imbalanced-diff}]
Let $t > 0$ be some parameter to be set later.  Then 
\begin{align*}
\E_{x \sim \mcl{D}_1}[ e^{-tx}] &=  \left(a/n \cdot e^{-t} + (1 - a/n)\right)^{\alpha n } \left(b/n \cdot e^{t} + (1 - b/n)\right)^{(1 - \alpha) n }  \\
\E_{x \sim \mcl{D}_2}[ e^{tx}] &=  \left(b/n \cdot e^{t} + (1 - b/n)\right)^{\alpha n } \left(a/n \cdot e^{-t} + (1 - a/n)\right)^{(1 - \alpha) n } 
\end{align*}
and thus
\begin{align*}
\Pr_{x \sim \mcl{D}_1}[x + K \leq  \theta ] \leq     \left(a/n \cdot e^{-t} + (1 - a/n)\right)^{\alpha n } \left(b/n \cdot e^{t} + (1 - b/n)\right)^{(1 - \alpha) n } e^{t(-K + \theta)}  \\
\Pr_{x \sim \mcl{D}_2}[x + K \geq -\theta ] \leq \left(b/n \cdot e^{t} + (1 - b/n)\right)^{\alpha n } \left(a/n \cdot e^{-t} + (1 - a/n)\right)^{(1 - \alpha) n } e^{t(K + \theta)} \,.
\end{align*}
We choose 
\[
e^{t} = \sqrt{\frac{a(1 - b/n)}{b(1 - a/n)}} \,.
\]
Note that $K$ is set precisely such that the RHS of the two expressions above are equal since this rearranges as
\[
e^{2tK} =   \left(\frac{a/n \cdot e^{-t} + (1 - a/n)}{b/n \cdot e^{t} + (1 - b/n)}\right)^{2\alpha - 1} =  \left( \frac{1 - a/n}{1 - b/n}\right)^{(2\alpha - 1)n } \,.
\]
Thus, we can instead multiply the two previous inequalities and get
\begin{align*}
\max\left(\Pr_{x \sim \mcl{D}_1}[x + K \leq  \theta ], \Pr_{x \sim \mcl{D}_2}[x + K \geq -\theta ]\right) &\leq \left(\left(a/n \cdot e^{-t} + (1 - a/n)\right) \left(b/n \cdot e^{t} + (1 - b/n)\right) \right)^{n/2}e^{t \theta} \\ &= \left( \sqrt{\frac{ab}{n^2}} + \sqrt{\left(1 - \frac{a}{n}\right)\left(1 - \frac{b}{n}\right)}\right)^{n} \left( \sqrt{\frac{a(1 - b/n)}{b(1 - a/n)}} \right)^{\theta} \\ &\leq \left(1 - \frac{(\sqrt{a} - \sqrt{b})^2}{n} \right)^{n/2} e^{\frac{\theta}{2} \log \frac{a(1 - b/n)}{b(1 - a/n)}} \\ &\leq e^{-C/2} \cdot e^{\frac{\theta}{2} \log R(a/n,b/n) }
\end{align*}
and this completes the proof.

\end{proof}

\section{Deferred Proofs from Section~\ref{sec:tail-bounds}}\label{appendix:tail-bounds}

\begin{proof}[Proof of Claim~\ref{claim:removed-column-sums}]
We first consider a fixed combinatorial subrectangle.  This is equivalent to considering an $n_1 \times n_2$ matrix $N$ with entries drawn from the same distribution.  Let $\Sigma$ be the sum of the entries of $N$.  We break into two cases.  First, if $\sigma \leq \frac{(n_1 + n_2)\sqrt{n}}{10^4n_1n_2}$ then for each entry of $N$, we have for any $t \geq 0$
\[
\E[e^{tN_{ij}}] \leq \E[1 + tN_{ij} + t^2N_{ij}^2(1 + e^{tN_{ij}})] \leq 1 + \sigma^2t^2(1 + e^t) \leq 1 + \sigma^2 e^{2t} \,.
\]
Thus using Markov's inequality and plugging in 
\[
t = \frac{1}{2}\log \left( \frac{(n_1 + n_2)\sqrt{n}}{4\sigma n_1 n_2}\right)
\]
we get
\begin{align*}
\Pr[\Sigma \geq (n_1 + n_2) \sigma \sqrt{n} ] &\leq \exp\left( n_1n_2 \sigma^2  \frac{(n_1 + n_2)\sqrt{n}}{4\sigma n_1 n_2}  - (n_1 + n_2) t\sigma \sqrt{n} \right) \\ & \leq \exp\left( -(n_1 + n_2)\sigma \sqrt{n}(t - 1) \right) \\& \leq e^{-10(n_1 + n_2) \log(n/n_1 + n/n_2)} 
\end{align*}
where the last step uses that $\sigma \geq 20/\sqrt{n}$.  Next, if $\sigma \geq \frac{(n_1 + n_2)\sqrt{n}}{10^4n_1n_2}$ then we have for any $0 \leq t \leq 1$,
\[
\E[e^{t N_{ij}}] \leq \E[1 + tN_{ij} + t^2N_{ij}^2]  \leq 1 + t^2\sigma^2 \,.
\]
Thus, by Markov's inequality and plugging in
\[
t = \frac{(n_1 + n_2 )\sqrt{n}}{10^4n_1n_2 \sigma }
\]
we get
\begin{align*}
\Pr[\Sigma \geq (n_1 + n_2) \sigma \sqrt{n} ] & \leq \exp\left( t^2\sigma^2 n_1 n_2 - t (n_1 + n_2) \sigma \sqrt{n} \right) \\& \leq \exp\left( - 10^{-4} n \right) \\& \leq e^{-10(n_1 + n_2) \log(n/n_1 + n/n_2)} 
\end{align*}
where in the last step, we use that $n_1,n_2 \leq 10^{-6}  n$.  Using the exact same argument, we get the same inequalities for $\Pr[\Sigma \leq (n_1 + n_2) \sigma \sqrt{n} ]$.  Now, combining everything and union bounding over all choices of a $n_1 \times n_2$ combinatorial rectangle gives a failure probability of at most
\[
2 e^{-10(n_1 + n_2) \log(n/n_1 + n/n_2)}  \binom{n}{n_1}  \binom{n}{n_2} \leq e^{-8(n_1 + n_2) \log(n/n_1 + n/n_2)}
\]
which completes the proof.
\end{proof}

\begin{proof}[Proof of Claim~\ref{claim:row-sum-distribution}]
Imagine sampling the matrix $A$ by independently drawing the entries $A_{ij}$ with $i < j$ and then filling in the remainder of the matrix symmetrically (with $0$ on the diagonal).  Note that $A_{ij}$ is drawn from $\text{Bernoulli}(a/n)$ if $i$ and $j$ are in the same community and $\text{Bernoulli}(b/n)$ otherwise. Note that to bound 
\[
\left\la \left(A - D(a/n, b/n)J \right) \odot L, \begin{bmatrix} x_1 \mathbf{1}^T \\ \vdots \\ x_n  \mathbf{1}^T \end{bmatrix}  \right\ra
\]
it suffices to consider when $x_1, \dots , x_n$ are indicators of some subset of size $\beta n$ since all other possibilities can be written as a convex combination of such choices.  Also, we may assume $\beta \geq 1/n$ (since at that point we are just considering single rows).  Finally, note that it suffices to consider when $\beta n$ is an integer because when it is not, we can consider subsets of size $\lfloor \beta n \rfloor$ and $\lceil \beta n \rceil$ and the RHS of the desired inequality is convex in $\beta$ and any choice of $x_1, \dots , x_n$ with $x_1 + \dots + x_n = \beta n$ can be written as a convex combination of indicator functions of sets of size  $\lfloor \beta n \rfloor$ and $\lceil \beta n \rceil$.

We will consider a fixed subset of $\beta n$ rows, say $S$ and then union bound over all possible choices. First, imagine that all entries in $S \times [n]$ are drawn independently (and diagonal entries are drawn from $\text{Bernoulli}(a/n)$).  We will apply  Claim~\ref{claim:binomials-imbalanced-diff} to bound the sum of all of the entries of $(A - D(a/n, b/n)J) \odot L$ indexed by $S \times n$.  Note that in Claim~\ref{claim:binomials-imbalanced-diff}, we are actually setting $n \leftarrow \beta n^2$, $a \leftarrow  a \beta n$ and $b \leftarrow b\beta n$.  We set 
\[
\theta = \frac{2\beta n( C/2 - \log 1/\beta - 3\kappa)}{\log R(a/n,b/n)} \,.
\]
We get that for the submatrix of $\left(A - D(a/n, b/n)J \right) \odot L$ indexed by $S \times n$, the sum of the entries is at least $\theta$ with probability at least 
\begin{equation}\label{eq:prob-bound1}
1 - e^{-\frac{\beta C n}{2} + \beta  n (  C/2 - \log 1/\beta - 3\kappa )  }  \geq  1 - e^{-\beta n (\log(1/\beta) + 1) - 2\kappa \beta n } \,.
\end{equation}

Of course, the above bound is assuming that all of the entries in $S \times [n]$ are drawn independently, which is not the case.  It remains to bound the difference from forcing $A$ to be symmetric and having $0$ on the diagonal.  It is clear that zeroing out the diagonal changes the desired quantity 
\[
\left\la \left(A - D(a/n, b/n)J \right) \odot L, \begin{bmatrix} x_1 \mathbf{1}^T \\ \vdots \\ x_n  \mathbf{1}^T \end{bmatrix}  \right\ra
\]
by at most $\beta n$.  

Next, we re-examine the entries indexed by $S \times [n]$.  All entries in $S \times n\backslash S$ are in fact drawn independently.  Now we can consider two possible ways to sample the entries in $S \times S$.  First, we sample the entries in $S \times S$ above the diagonal and symmetrically fill in the entries below the diagonal.  This is equivalent to the sampling process for $A$.  Alternatively, we can sample the entries below the diagonal independently as well, which is equivalent to the fully independent process considered previously.  It remains to upper bound the difference between the two sampling processes.  The difference has mean $0$ and is a sum/difference of independent random variables that take values in $\{0,1 \}$ and are nonzero with probability at most $a/n$.  Let this difference be $\Delta$.  By Claim~\ref{claim:removed-column-sums}, we have
\begin{equation}\label{eq:prob-bound2}
\Pr[|\Delta| \geq \beta n \sqrt{a} ] \leq  e^{-8\beta n \log( 1/\beta) } \leq \frac{1}{n^4  \binom{n}{\beta n}} \,.
\end{equation}

Finally, we can simply union bound the combination of \eqref{eq:prob-bound1} and  \eqref{eq:prob-bound2} over all choices of a subset $S \subset n$ of size $\beta n$ to get the desired inequality.
\end{proof}

\begin{proof}[Proof of Lemma~\ref{lem:row-sum-lowerbound}]
First we deal with the case where $x_1 + \dots  +x_n \leq 100$.  We bound the probability that the sum of the entries in each row of  $ \left(A - D(a/n, b/n)J \right) \odot L$ is at least the quantity
\[
T = K\sqrt{a + b}  - \frac{\theta n \sqrt{a + b} + \frac{10^4(\kappa -\sqrt{C})}{\log R(a/n, b/n)} }{100} \,.
\]
Note that if this happens, then we clearly get the desired inequality for all $x_1 + \dots  +x_n \leq 100$.  By Claim~\ref{claim:binomials-imbalanced-diff} and a union bound, all rows of $ \left(A - D(a/n, b/n)J \right) \odot L$ have sum of entries at least $T$  with probability  at least
\begin{equation}\label{eq:prob-dumb-case}
1 - n e^{\left(-\frac{C}{2} + \frac{T}{2} \log R(a/n,b/n) \right) } \,.
\end{equation}
Now we examine the expression in the exponent.  By definition, it is
\begin{align*}
 -\frac{C}{2} + \frac{T}{2} \log R(a/n,b/n) &= -\frac{C}{2} + \frac{K \sqrt{a + b}}{2} \log R(a/n,b/n) - \frac{\theta n \sqrt{a + b}}{200} \log R(a/n,b/n) - \frac{10^4(\kappa -\sqrt{C})}{200} \\ & \leq \log \theta - \frac{\theta n \sqrt{a + b}}{200} \log R(a/n,b/n)  - \frac{K \sqrt{a + b}}{2} \log R(a/n,b/n) - 50(\kappa -\sqrt{C}) \\ &\leq -\log \frac{n \sqrt{a + b} \log R(a/n,b/n)}{200} - \frac{K \sqrt{a + b}}{2} \log R(a/n,b/n) - 50(\kappa -\sqrt{C}) \,.
\end{align*}
Note that by definition
\[
\log R(a/n,b/n) \geq \log(a/b) = 2\log\sqrt{a/b} \geq 2 \frac{\sqrt{a} - \sqrt{b}}{\sqrt{a + b}} = \frac{2\sqrt{C}}{\sqrt{a + b}} \,.
\]
Thus, substituting back into  \eqref{eq:prob-dumb-case}   the probability in question is at least
\[
1 - n e^{\left(-\frac{C}{2} + \frac{T}{2} \log R(a/n,b/n) \right) } \geq 1 - e^{-50\kappa} \,.
\]

Now it remains to consider when $x_1 + \dots + x_n \geq 100$.  First apply Claim~\ref{claim:row-sum-distribution} and union bound over all $\beta$ such that $\beta$ is an integer and $\beta n \geq 100$.  Since the RHS of the desired inequality is linear in $x_1 + \dots  +x_n$, it suffices to consider when $x_1, \dots , x_n$ are all $0$ or $1$ (since all other possibilities can be written as a suitable linear combination).  Now let $x_1 + \dots + x_n = \beta n$.  It suffices to compare the RHS of the last expression in Claim~\ref{claim:row-sum-distribution} to the RHS of the desired inequality above.  As before, define 
\[
Q = C/2 - \log(1/\beta) - 3\kappa\,.
\]
If $\beta \leq \theta/(2K)$ then
\begin{align*}
2\beta n\left(\frac{Q}{\log R(a/n, b/n)} - \sqrt{a + b}\right) &= -2\beta n \sqrt{a + b} + \frac{2 n \left( \beta( C/2 -  3\kappa ) + \beta \log \beta \right)}{\log R(a/n, b/n)} \\ &\geq -2\beta n \sqrt{a + b} -  \frac{2 n}{\log R(a/n, b/n)} e^{-C/2 + 3\kappa } \\ &\geq  -2\beta n \sqrt{a + b} - n \frac{\sqrt{a + b}}{\sqrt{C}} e^{-C/2 + 3\kappa }  \\ &\geq -2(\beta + 0.1\theta) n \sqrt{a + b}  \\ &\geq (K \beta - \theta)n \sqrt{a + b}
\end{align*}
and we are done by the guarantees in Claim~\ref{claim:row-sum-distribution}.  Otherwise, if $\beta \geq \theta/(2K)$ then we also must have 
\[
Q \geq K \sqrt{a + b} \log R(a/n,b/n) - \log(2K)
\]
and thus 
\[
2\beta n\left(\frac{Q}{\log R(a/n, b/n)} - \sqrt{a + b}\right) \geq 2\beta n \sqrt{a + b} \left( K - (\log(2K) +1 )\right) \geq K\beta n \sqrt{a + b}
\]
and again we are done by the guarantees in Claim~\ref{claim:row-sum-distribution}.  This completes the proof.
\end{proof}

\begin{proof}[Proof of Corollary~\ref{coro:resolvable2}]

By Lemma~\ref{lem:row-sum-lowerbound} we have that with probability at least $1 - e^{-10\kappa} - 1/n^2$ the following inequality holds for all $0 \leq x_1, \dots , x_n \leq 1$ with $x_1 + \dots  + x_n \leq 10^{-6}n$:
\begin{equation}\label{eq:base-resolve}
\left\la \left(A - D(a/n, b/n)J \right) \odot L, \begin{bmatrix} x_1 \mathbf{1}^T \\ \vdots \\ x_n  \mathbf{1}^T \end{bmatrix}  \right\ra \geq  (K(x_1 + \dots + x_n) - \theta n )  \sqrt{a + b}  - \frac{\max(0, 10^4(\kappa -\sqrt{C}))}{\log R(a/n,b/n)} \,. 
\end{equation}
Now we need to consider what happens when we replace $\left(A - D(a/n, b/n)J \right) \odot L$ in the above with its restriction to $S \times S$.  Since the condition of resolvability is linear in the sum $x_1 + \dots + x_n$, it suffices to consider when all of the $x_i$ are $0$ or $1$.  Also, it suffices to consider when the only nonzero $x_i$ correspond to $i \in S$.  We can first write
\begin{align*}
\left\la \left(A - D(a/n, b/n)J \right) \odot L, \begin{bmatrix} x_1 \mathbf{1}^T \\ \vdots \\ x_n  \mathbf{1}^T \end{bmatrix}  \right\ra &= \left\la \left(A - \frac{a + b}{2n}J - \frac{a - b}{2n}L \right) \odot L, \begin{bmatrix} x_1 \mathbf{1}^T \\ \vdots \\ x_n  \mathbf{1}^T \end{bmatrix}  \right\ra  \\ &\quad + \left\la \left(\frac{a + b}{2n}J + \frac{a - b}{2n}L - D(a/n, b/n)J \right) \odot L, \begin{bmatrix} x_1 \mathbf{1}^T \\ \vdots \\ x_n  \mathbf{1}^T \end{bmatrix}  \right\ra \,.
\end{align*}
Now, for the second term, note that by Claim~\ref{claim:log-bound}, all entries of the matrix $\frac{a + b}{2n}J + \frac{a - b}{2n}L - D(a/n, b/n)J $ are in the interval $\left[ -\frac{a - b}{n}, \frac{a - b}{n} \right] $.  Thus, when restricting to $S \times S$, the second term changes by at most
\[
\frac{a - b}{n} \cdot \frac{K}{10\sqrt{C}}n(x_1 + \dots + x_n) \leq 0.2K\sqrt{a + b}(x_1 + \dots  +x_n) \,.
\]
Now, to bound the first term, we can apply Claim~\ref{claim:removed-column-sums}.  Note that the entries of the matrix 
\[
A - \frac{a + b}{2n}J - \frac{a - b}{2n}L
\]
are drawn from a distribution with mean $0$.  Also, for sub-rectangles of $[n] \times [n]$ indexed by disjoint sets $T_1, T_2$, the entries in that subrectangle are all independent.  Thus, we can union bound Claim~\ref{claim:removed-column-sums} over all such sub-rectangles with  $|T_1|, |T_2| \leq 10^{-6}n$ resulting in a failure probability of at most 
\begin{align*}
\sum_{n_1 = 1}^{10^{-6}n}\sum_{n_2 = 1}^{10^{-6}n} e^{-8(n_1 + n_2)\log(n/n_1 + n/n_2)} \leq \frac{1}{n^2} \,. 
\end{align*}
Assuming that the result in Claim~\ref{claim:removed-column-sums} holds for all such sub-rectangles, when restricting to $S \times S$ the term
\[
 \left\la \left(A - \frac{a + b}{2n}J - \frac{a - b}{2n}L \right) \odot L, \begin{bmatrix} x_1 \mathbf{1}^T \\ \vdots \\ x_n  \mathbf{1}^T \end{bmatrix}  \right\ra 
\]
can change by at most
\[
(n - |S| + x_1 + \dots  +x_n)\sqrt{a} \,.
\]
Combining our bounds on the change when restricting to $S \times S$ with \eqref{eq:base-resolve}, we get that
\begin{align*}
&\left\la \left(\left(A - D(a/n, b/n)J \right) \odot L \right)_{S \times S}, \begin{bmatrix} x_1 \mathbf{1}^T \\ \vdots \\ x_n  \mathbf{1}^T \end{bmatrix}  \right\ra \\ &\geq  (K(x_1 + \dots + x_n) - \theta n )  \sqrt{a + b}  - \frac{\max(0, 10^4(\kappa -\sqrt{C}))}{\log R(a/n,b/n)} \\ &\quad - 0.2K\sqrt{a + b}(x_1 + \dots  +x_n) - (n - |S| + x_1 + \dots  +x_n)\sqrt{a} \\ &\geq \left(0.5K(x_1 + \dots + x_n) - \theta n - (n - |S|) \right)  \sqrt{a + b}  - \frac{\max(0, 10^4(\kappa -\sqrt{C}))}{\log R(a/n,b/n)} \,.
\end{align*}
This immediately implies the desired property (note that the factor of $1.1$ in the second parameter is because we are now considering an $|S| \times |S|$ matrix instead of an $n \times n$ matrix but $1.1|S| \geq n$).
\end{proof}

\section{Deferred Proofs from Section~\ref{sec:initialization}}\label{appendix:initialization}

\begin{proof}[Proof of Lemma~\ref{lem:initial-SDP-analysis}]
Let $A_0$ be the pure SBM adjacency matrix i.e. before the semi-random and adversarial corruptions are added.  Let $L$ be the matrix whose entries are $L_{ij} = 1$ if $i,j$ are in the same community and $L_{ij} = -1$ otherwise.  By Corollary~\ref{coro:spectral-remove} with probability at least $1 - 2/n^2$, there is a subset $S \subset [n]$ of size at least $(1 - e^{-2C})n$ such that 
\begin{equation}\label{eq:spectral-bound}
\norm{\left(A_0 - \frac{a + b}{2n}J - \frac{a - b}{2n}L \right)_{S \times S}}_{\op} \leq \chi \sqrt{a + b} \,.
\end{equation}

Let $T$ be the  subset of $S$ consisting of the uncorrupted nodes.  Note that $|T| \geq (1 - e^{-2C} - \eps) n$.  Let $v \in \R^n$ be the indicator vector of the set $T$ and for $i = 1,2, \dots , k$, let $v_i \in \R^n$ be the vector that has $1$ in entries indexed by elements of $S_i \cap T$ and $0$ in other entries.  Define
\[
N = vv^T - \sum_{i = 1}^k v_iv_i^T  \,.
\]
Note that alternatively, we can obtain $N$ by taking $(J - L)/2$ and zeroing out all entries except for those in $T \times T$.  Now we can write
\begin{align*}
(A  - (a/n)J - F)_{T \times T} = \left(A_0 - \frac{a + b}{2n}J - \frac{a - b}{2n}L  - \frac{a - b}{2n}(J - L) + E - F \right)_{T \times T} \\ = \left(\left( A_0 - \frac{a + b}{2n}J - \frac{a - b}{2n}L  \right) - \frac{a - b}{n}N + E - F \right)_{T \times T}
\end{align*}
where $E$ is a matrix that has the same signs as $L$, corresponding to the semi-random noise.  Also, note that
\[
\norm{N \odot W}_1 \leq \norm{(vv^T) \odot W }_1 +  \norm{  \left(\sum_{i = 1}^k v_iv_i^T \right) \odot W}_1 \leq 2\norm{W}_1 = 2n \,.
\]
Thus, using \eqref{eq:spectral-bound}, we have
\[
\langle (A  - (a/n)J - F)_{T \times T} \odot W, N \rangle \leq 2\chi n \sqrt{a + b} - \frac{a - b}{n} \langle N, W \rangle 
\]
since $E$ must be non-positive on the entries where $N$ is supported.  Alternatively, by the spectral constraint in the SDP, we must have
\[
\langle (A  - (a/n)J - F)_{T \times T} \odot W, N \rangle \geq -2\chi n \sqrt{a + b} \,.
\]
Thus, we conclude 
\[
\langle N, W \rangle \leq  4\chi n^2 \frac{\sqrt{a + b}}{a - b} \leq \frac{4\chi n^2}{\sqrt{C}} \,.
\]
The above is the sum of the entries of $W$ over the support of $N$.  Finally, there are at most $2(e^{-2C} + \eps)n^2$ additional entries $W_{ij}$ where $i,j$ are in different communities that we need to consider.  Thus, overall, the sum of the entries $W_{ij}$ where $i,j$ are in different communities is at most $(5\chi/\sqrt{C} +2\eps ) n^2$.

Next, to prove the second property, we will construct a feasible solution and argue about its objective value.  Let
\[
\wt{W} = \sum_{i = 1}^k v_iv_i^T \,.
\]
We can choose $F$ such that $ (E - F) \odot \wt{W} = 0$ (where $E$ denotes the semi-random noise).  We then have
\begin{align*}
(A  - (a/n)J - F) \odot \wt{W} = \left(A_0 - \frac{a + b}{2n}J - \frac{a - b}{2n}L  - \frac{a - b}{2n}(J - L) + E - F \right) \odot \wt{W} \\ =  \left(A_0 - \frac{a + b}{2n}J - \frac{a - b}{2n}L\right) \odot \wt{W}
\end{align*}
and the feasibility now follows immediately from \eqref{eq:spectral-bound}.  Finally, note that the objective value of $\wt{W}$ is at least $\sum_{i = 1}^k|S_i|^2 - 2(e^{-2C} + \eps)n^2$.  Combining this with the first property and the optimality of $W$ gives the desired bound. 
\end{proof}

\begin{proof}[Proof of Lemma~\ref{lem:rough-clustering}]
First, we upper bound the $k$-means error of the true labelling.  We can also set the means to be $\mu_1, \dots , \mu_k$ where $\mu_i$ is the indicator vector of $S_i$ (since moving the means will only increase the error).  Let $W'$ be the matrix with $W_{ij}' = 1$ if $i$ and $j$ are in the same community (in the base SBM) and $W_{ij}' = 0$ otherwise.  Note that Lemma~\ref{lem:initial-SDP-analysis} implies that
\begin{equation}\label{eq:frob-bound}
\norm{W - W'}_F^2 \leq (11\chi/\sqrt{C} + 6\eps)n^2
\end{equation}
since all entries are between $0$ and $1$ and thus the above also bounds the optimal value of the $k$-means objective i.e.
\[
\text{k-means}_{\textsf{opt}}(W) \leq (11\chi/\sqrt{C} + 6\eps)n^2 \,.
\]
Now consider any alternative clustering into sets $S_1', \dots , S_k'$.  Also, first imagine replacing $W$ with $W'$ and evaluating the objective with respect to the rows of $W'$.  Consider the set $S_1'$ and let $s_1 = |S_1' \cap S_1| , \dots,  s_k = |S_1' \cap S_k|$.  Note that the $k$-means objective on $S_1'$ is at least 
\begin{equation}\label{eq:kmeans-lowerbound}
(s_1 + \dots + s_k - \max(s_1, \dots , s_k)) \cdot \frac{\alpha n}{4k} \,.
\end{equation}
Now we will combine the above with the guarantees of the $k$-means approximation algorithm to prove that the clustering that is actually computed has small error.  Let the computed clustering be $S_1', \dots , S_k'$.  Now define the two quantities
\begin{align*}
\delta = \frac{1}{n} \min_{\substack{\pi:[k] \rightarrow [k] \\ \pi \text{ invertible}}} \left( \sum_{i = 1}^k |S_i' \backslash S_{\pi(i)}| \right) \\
\delta' = \frac{1}{n} \min_{f:[k] \rightarrow [k]} \left( \sum_{i = 1}^k |S_i' \backslash S_{f(i)}| \right)
\end{align*}
where in the above the difference is that $\pi$ must be a permutation but $f$ may be an arbitrary function.  Note that $\delta$ is exactly the error of the clustering $S_1', \dots , S_k'$ relative to the ground truth $S_1, \dots , S_k$.  Now we relate $\delta'$ to $\delta$.  Let $f: [k] \rightarrow [k]$ be the function that achieves the value of $\delta'$.  Let $r$ be the number of elements in the range of $f$.  Then clearly $\delta' \geq \alpha (k - r)/k$.  Now consider modifying $f$ into a permutation by changing its value on exactly $k - r$ inputs.  This affects the value of   
\[
\sum_{i = 1}^k |S_i' \backslash S_{f(i)}| 
\]
by at most $(k - r)n/(\alpha k)$.  Thus, we must have
\[
\delta \leq \delta' (1 + 1/\alpha^2) \,. 
\]
Next, note that \eqref{eq:kmeans-lowerbound} implies
\[
\text{k-means}_{S_1', \dots , S_k'}(W') \geq \frac{\delta' \alpha n^2}{4k} \geq \frac{\alpha  \delta n^2}{4(1 + 1/\alpha^2) k} \geq \frac{\alpha^3\delta n^2}{8k}
\]
and combining with \eqref{eq:frob-bound} and using Cauchy-Schwarz gives
\[
\text{k-means}_{S_1', \dots , S_k'}(W) \geq  \left(\frac{1}{2} \frac{\alpha^3 \delta}{8k}  - (11\chi/\sqrt{C} + 6\eps) \right) n^2 \,.
\]
Thus, by the guarantees of our $k$-means approximation algorithm, we must have
\[
\delta \leq 11 \cdot (11\chi/\sqrt{C} + 6\eps) \cdot \frac{16 k}{\alpha^3} \leq \frac{10^4  k}{\alpha^3}\left( \frac{\chi}{\sqrt{C}} + \eps \right) 
\]
as desired.
\end{proof}

\input{Z2}

%% file: Z2.tex
\section{Robust $\Z_2$ Synchronization}\label{sec:Z2}

In this section, we prove our results for $\Z_2$-synchronization.  The proof follows essentially the same steps as the proof for community detection.  We first prove bounds on row sums and resolvability in Section~\ref{sec:Z2-rowsums}.  We then give our initialization algorithm in Section~\ref{sec:Z2-init}.  Finally, we give our boosting algorithm in Section~\ref{sec:Z2-boosting} and complete the proof of Theorem~\ref{thm:main-sync}.  Many of the tools in Section~\ref{sec:2community-SDP} can be re-used directly to analyze our boosting procedure for $\Z_2$-synchronization.

\subsection{Key Properties of Spiked Random Matrix Models}\label{sec:Z2-rowsums}

\begin{lemma}\label{lem:Z2-row-sumlowerbound}
Let $\ell \in \{-1,1 \}^n$ be a sign vector and let $L = \ell \ell^T$.  Let $A$ be matrix generated as
\[
A = \lambda \ell \ell^T /\sqrt{n} +  E
\]
where $E$ has i.i.d. entries drawn from $N(0,1)$.  Let $\kappa, K$ be some parameters and assume $\lambda,  K \geq 10^2$.  Let
\[
\theta = e^{-\lambda^2/2 + 2K \lambda + 3\kappa} \,.
\]
Then with probability at least $1 - e^{-10\kappa}$, we have for any $0 \leq x_1, \dots , x_n \leq 1$ with $x_1 + \dots + x_n \leq 0.1n$ that 
\[
\left\la A \odot L , \begin{bmatrix} x_1 \mathbf{1}^T \\ \vdots \\ x_n  \mathbf{1}^T \end{bmatrix} \right\ra \geq (K(x_1 + \dots + x_n) - \theta n ) \cdot \sqrt{n} - 10^4 \sqrt{n} \max(0, \kappa - \lambda)
\]
\end{lemma}
\begin{proof}
Note that the RHS is linear in the $x_i$ so it suffices to consider when the $x_i$ are the indicators of some set $S$.  Let $|S| = \beta n$ and now fix this set $S$.  First, we deal with the case where $|S| \leq 10$.  For this case, we will simply lower bound all row-sums of $A \odot L$.  Each row sum of $A \odot L$ is distributed as $N( \lambda \sqrt{n} , n) $.  We now bound the probability that all row sums are at least 
\[
K \sqrt{n} - \frac{\theta n \sqrt{n}}{10} - 10^3 \sqrt{n} \max(0, \kappa - \lambda)
\]
and this will then finish the case where $|S| \leq 10$.  By a naive union bound, this probability is at least
\begin{equation}\label{eq:failure-expression}
1 - n \exp\left(  - \frac{ ( \lambda - K)^2 n +  0.01\theta^2 n^3 + 10^6 n \max(0, \kappa - \lambda)^2}{2 n}\right) \,.
\end{equation}
The expression in the exponent above is 
\begin{align*}
\frac{ ( \lambda - K)^2 n +  0.01\theta^2 n^3 + 10^6 n \max(0, \kappa - \lambda)^2}{2 n} \geq \frac{-\log(\theta^2)}{2} + \frac{\theta^2 n^2}{200} + K\lambda + 10^5 \max(0, \kappa - \lambda)^2 \\ \geq \log n + 50\kappa 
\end{align*}
and thus the expression in \eqref{eq:failure-expression} is at least $1 - e^{-50\kappa}$, which finishes this case.

Now, we consider the case where $|S| > 10$.  The distribution of the inner product on the LHS in the statement of the claim is 
\[
N\left( \beta \lambda n^{3/2}, \beta n^2 \right) 
\]
and thus, the probability that the desired inequality fails is at most
\begin{align*}
\exp\left(-\frac{ ((\lambda - K)\beta + \theta)^2 n^3 }{2 \beta n^2}\right) \leq \exp\left( -\frac{(\lambda - K)^2 \beta n}{2}  - \frac{\theta^2 n }{2\beta } \right) \leq \exp\left( \frac{(\log \theta^2 - 6\kappa) \beta n}{2}  - \frac{\theta^2 n }{2\beta } \right) \\ \leq \exp\left(-\frac{ \beta  n (2\log(1/\beta) + 1)}{2} - 3\kappa \beta n  \right) \leq \exp\left( -\beta n (\log(1/\beta) + 1) - 30 \kappa \beta n \right) \,.
\end{align*}
Simply union bounding the above over all $\beta$ such that $\beta n $ is an integer completes the proof.
\end{proof}

We need a few additional tail bounds that follow from standard concentration inequalities.
\begin{claim}\label{claim:Gaussian-rectangles}
Let $n, n_1, n_2$ be parameters with $n_1, n_2 \leq 0.1n$.  Let $E$ be an $n \times n$ matrix whose entries are drawn i.i.d. from $N(0,1)$.  Then with probability at least $1 - 1/n^5$, the magnitude of the sum of the entries over any $n_1 \times n_2$ combinatorial subrectangle of $E$ is at most $(n_1 + n_2)\sqrt{n}$.
\end{claim}
\begin{proof}
Note that for a fixed rectangle, the distribution of the sum is $N(0, \sqrt{n_1n_2})$.  There are 
\[
\binom{n}{n_1} \binom{n}{n_2} \leq e^{n_1 (\log(n/n_1) + 1) + n_2 (\log(n/n_2) + 1)}
\]
distinct rectangles, so using standard Gaussian tail bounds and union bounding, we get the desired result.
\end{proof}

We also have a standard spectral bound from random matrix theory.
\begin{claim}[See \cite{erdos2011universality} ]\label{claim:Z2-spectral}
Let $E$ be a matrix with i.i.d. entries drawn from $N(0,1)$.  Then with probability $1 - 1/n^2$, we have
\[
\norm{E }_{\op} \leq 3\sqrt{n} \,.
\]
\end{claim}

We can now prove resolvability for the matrix $A \odot L$ generated in $\Z_2$-synchronization and all of its submatrices.  This is the analog of Corollary~\ref{coro:resolvable2}.  
\begin{corollary}\label{coro:Z2-resolvable}
Let $\ell \in \{-1,1 \}^n$ be a sign vector and let $L = \ell \ell^T$.  Let $A$ be matrix generated as
\[
A = \lambda \ell \ell^T /\sqrt{n} +  E
\]
where $E$ has i.i.d. entries drawn from $N(0,1)$.  Let $\kappa , K $ be parameters and assume $K \geq 10^2, \lambda \geq 10^2K$.  Then with probability $1 - e^{-10\kappa} - 2/n^2$, we have that for all subsets $S \subset [n]$ with $|S| \geq (1 - K/(3\lambda)) n $ that $(A \odot L)_{S \times S}$ is resolvable with parameters
\[
\left(0.5 K\sqrt{n},  1.1\sqrt{n}\left(\theta + \frac{n - |S|}{n}  + \frac{10^4 \max(0, \kappa - \lambda)}{n} \right)\right)
\]
where 
\[
\theta =  e^{-\lambda^2/2 + 3\kappa + 2K \lambda } \,.
\]
\end{corollary}
\begin{proof}
By Lemma~\ref{lem:Z2-row-sumlowerbound} we have that with probability at least $1 - e^{-10\kappa}$ the following inequality holds for all $0 \leq x_1, \dots , x_n \leq 1$ with $x_1 + \dots  + x_n \leq 0.1n$:
\[
\left\la A \odot L, \begin{bmatrix} x_1 \mathbf{1}^T \\ \vdots \\ x_n  \mathbf{1}^T \end{bmatrix}  \right\ra \geq  (K(x_1 + \dots + x_n) - \theta n )  \sqrt{n}  - 10^4 \sqrt{n} \max(0, \kappa - \lambda) \,. 
\]
Now we need to consider what happens when we replace $A \odot L$ in the above with its restriction to $S \times S$.  Since the condition of resolvability is linear in the sum $x_1 + \dots + x_n$, it suffices to consider when all of the $x_i$ are $0$ or $1$.  Also, it suffices to consider when the only nonzero $x_i$ correspond to $i \in S$.  We can first write
\begin{align*}
\left\la A \odot L, \begin{bmatrix} x_1 \mathbf{1}^T \\ \vdots \\ x_n  \mathbf{1}^T \end{bmatrix}  \right\ra = \left\la E \odot L , \begin{bmatrix} x_1 \mathbf{1}^T \\ \vdots \\ x_n  \mathbf{1}^T \end{bmatrix}  \right\ra  + \left\la  \frac{\lambda L}{\sqrt{n}} \odot L , \begin{bmatrix} x_1 \mathbf{1}^T \\ \vdots \\ x_n  \mathbf{1}^T \end{bmatrix}  \right\ra \,.
\end{align*}
Now, by Claim~\ref{claim:Gaussian-rectangles}, the first term changes by at most $(n - |S| + x_1 + \dots + x_n) \sqrt{n}$ when restricting to $S \times S$.  The second term changes by at most 
\[
(n - |S|)(x_1 + \dots + x_n)\lambda/\sqrt{n} \leq K \sqrt{n}(x_1 + \dots + x_n)/3
\]
when restricting to $S \times S$.  Thus,
\begin{align*}
\left\la \left( A \odot L \right)_{S \times S}, \begin{bmatrix} x_1 \mathbf{1}^T \\ \vdots \\ x_n  \mathbf{1}^T \end{bmatrix}  \right\ra  \geq (K(x_1 + \dots + x_n) - \theta n )  \sqrt{n}  - 10^4 \sqrt{n} \max(0, \kappa - \lambda)  \\ - (n - |S| + x_1 + \dots + x_n) \sqrt{n} - K \sqrt{n}(x_1 + \dots + x_n)/3
\\ \geq 0.5K(x_1 + \dots + x_n)\sqrt{n} -  \left(\theta n  + (n - |S|) + 10^4 \max(0, \kappa - \lambda) \right)\sqrt{n}
\end{align*}
This immediately implies the desired property (note that the factor of $1.1$ in the second parameter is because we are now considering an $|S| \times |S|$ matrix instead of an $n \times n$ matrix but $1.1|S| \geq n$).
\end{proof}

\subsection{Initialization}\label{sec:Z2-init}

The initialization step involves solving a very similar SDP to that in Section~\ref{sec:initialization}.  We formulate the variant for $\Z_2$-synchronization below.
\begin{definition}[Initialization SDP  ($\Z_2$-Synchronization)]
Assume that we are given some $n \times n$ matrix $A$  and parameter $\lambda$.  We solve for matrices $W,D$ such that 
\begin{align*}
    &0 \leq W_{ij}, D_{ij} \leq 1 \quad \forall i,j \\
    &\norm{W}_1 \leq n \\
    &-3\sqrt{n} I \preceq (A - \lambda J/\sqrt{n} - D) \odot W \preceq 3\sqrt{n} I 
\end{align*}
\end{definition}

Since there are only two communities, we can use a simpler post-processing algorithm after we get our solution $W$ to the initialization SDP.  Instead of using $k$-means, we simply label according to the signs of the top right singular vector of $W - J/2$.
\begin{algorithm}[H]
\caption{{\sc Compute Initial Labelling ($\Z_2$-Synchronization)} }
\begin{algorithmic}
\State \textbf{Input:} Adjacency matrix $A$, parameter $\lambda$
\State Run Initialization SDP to obtain solution $W$
\State Let $v$ be the top right singular vector of $W - J/2$
\State Let $\ell_{\text{init}} \in \{-1,1 \}^n$ be the signs of the vector $v$
\State \textbf{Output:} $\ell_{\text{init}}$
\end{algorithmic}
\label{alg:Z2-round-initialSDP}
\end{algorithm}

The analysis of the initialization SDP is essentially the same as before.  We first show that the optimal solution $W$ is close to the intended solution which has entries $ W_{ij} = 1$ for $i$ and $j$ with the same label and $W_{ij} = 0$ otherwise.

\begin{lemma}\label{lem:Z2-initial-SDP-analysis}
Let $A$ be a matrix from an $\eps$-corrupted $\Z_2$-synchronization instance with parameter $\lambda$.  Let the sizes of the two groups (in generating $A$)  be $m, n-m$.  Then with probability at least $1 - 2/n^2$, the solution $W$ to the Initialization SDP  has the following properties:
\begin{itemize}
    \item The sum of the entries of $W_{ij}$ where $i,j$ have different labels is at most $(6/\lambda +2\eps ) n^2$
    \item The sum of the entries of $W_{ij}$ where $i,j$ have the same labels is at least 
    \[
    m^2 + (n-m)^2 - (6/\lambda  + 4\eps)n^2 \,. 
    \]
\end{itemize}
\end{lemma}
\begin{proof}
Let $A_0$ be the pure $\Z_2$-synchronization matrix i.e. before the semi-random and adversarial corruptions are added.  Let $\ell$ be the true labelling and let $L = \ell \ell^T$.  Let $S$ be the  subset of uncorrupted nodes.  Note that $|S| \geq (1  - \eps) n$.  Let $N$ be the matrix obtained by taking $(J - L)/2$ and zeroing out all entries except for those in $S \times S$.  Now we can write
\begin{align*}
(A  - \lambda J/\sqrt{n} - D)_{S \times S} = \left(A_0 - \lambda L/\sqrt{n}  - \lambda (J - L)/\sqrt{n} + F - D \right)_{S \times S} \\ = \left(E - \frac{2\lambda}{\sqrt{n}}N + F - D \right)_{S \times S}
\end{align*}
where $E$ has i.i.d. standard Gaussian entries and $F$ is a matrix that has the same signs as $L$, corresponding to the semi-random noise.  Also, note that
\[
\norm{N \odot W}_1 \leq  2\norm{W}_1 \leq  2n \,.
\]
Thus, using Claim~\ref{claim:Z2-spectral} for a spectral bound on $E$, we have
\[
\langle (A  - \lambda J/\sqrt{n} - D)_{S \times S} \odot W, N \rangle \leq 6n \sqrt{n} - \frac{2\lambda}{\sqrt{n}} \langle N, W \rangle 
\]
since $F$ must be non-positive on the entries where $N$ is supported.  Alternatively, by the spectral constraint in the SDP, we must have
\[
\langle (A  - \lambda J/\sqrt{n} - D)_{S \times S} \odot W, N \rangle \geq -6 n \sqrt{n} \,.
\]
Thus, we conclude 
\[
\langle N, W \rangle \leq \frac{6n^2}{\lambda}\,.
\]
The above is the sum of the entries of $W$ over the support of $N$.  Finally, there are at most $2\eps n^2$ additional entries $W_{ij}$ where $i,j$ have different labels that we need to consider.  Thus, overall, the sum of the entries $W_{ij}$ where $i,j$ have different labels is at most $(6/\lambda +2\eps ) n^2$.

Next, to prove the second property, we will construct a feasible solution and argue about its objective value.  Let $\wt{W}$ be the matrix obtained by taking $(J + L)/2$ and zeroing out all entries outside $S \times S$.  We can choose $D$ such that $ (F - D) \odot \wt{W} = 0$ (where $F$ denotes the semi-random noise).  We then have
\begin{align*}
(A  - \lambda J/\sqrt{n} - D) \odot \wt{W} = \left(A_0 - \lambda L/\sqrt{n}  - \lambda (J - L)/\sqrt{n} + F - D  \right) \odot \wt{W} \\ =  \left(A_0 - \lambda L/\sqrt{n} \right) \odot \wt{W}
\end{align*}
and the feasibility now follows immediately from the spectral bound in Claim~\ref{claim:Z2-spectral}.  Finally, note that the objective value of $\wt{W}$ is at least $m^2 + (n-m)^2 - 2\eps n^2$.  Combining this with the first property and the optimality of $W$ gives the desired bound. 
\end{proof}

Now we can prove that the initialization algorithm indeed computes a good initial labelling.
\begin{lemma}\label{lem:Z2-rough-clustering}
Let $A$ be a matrix from an $\eps$-corrupted $\Z_2$-synchronization instance with parameter $\lambda$.  Then with probability at least $1 - 2/n^2$, the output of {\sc Compute Initial Labelling ($\Z_2$-synchronization)}  has error at most
\[
10^3(1/\lambda + \eps) \,.
\]
\end{lemma}
\begin{proof}
Let the true labelling of the nodes be given by $\ell \in \{-1 , 1 \}^n$ and let $L  = \ell \ell^T$.  Define the matrix
\[
W' = (J + L)/2 \,.
\]
Note that Lemma~\ref{lem:Z2-initial-SDP-analysis} implies that 
\[
\norm{W - W'}_F^2 \leq (12/\lambda + 6\eps)n^2 \,.
\]
Note that the matrix $W' - J/2 = \ell \ell^T / 2$ is rank-$1$ and has top eigenvalue $n/2$.  Now consider the top left and right singular vector of $W - J/2$, say $u,v$ respectively, and let the top singluar value be $\rho$.  Then we must have
\[
\norm{W - J/2 - \rho u v^T }_F  \leq \norm{W - W'}_F \,.
\]
Thus, we must have
\[
\norm{\ell \ell^T / 2 -  \rho u v^T}_F  \leq 2 \norm{W - W'}_F \leq 2 \sqrt{ 12/\lambda + 6\eps} n \,.
\]
Now let $c$ be the length of the projection of the unit vector $\ell/\sqrt{n}$ onto the orthogonal complement of $v$.  Then 
\[
\norm{\ell \ell^T / 2 -  \rho u v^T}_F \geq cn/2
\]
so $c \leq 4 \sqrt{ 12/\lambda + 6\eps}$.  Now WLOG $\la \ell , v \ra > 0$ (since otherwise, we can simply negate $\ell$).  Then we must have $\norm{\ell/\sqrt{n} - v} \leq 8\sqrt{ 12/\lambda + 6\eps}$.  This immediately implies that $v$ and $\ell$ have opposite signs on at most $ 10^3(1/\lambda + \eps)n$ entries and we are done.
\end{proof}

\subsection{Boosting} \label{sec:Z2-boosting}

Finally, we complete the proof of Theorem~\ref{thm:main-sync}.  For the boosting step, we can actually solve the same boosting SDP as community detection (for two communities) except with a different setting of parameters.  We can also use Theorem~\ref{thm:boosting-SDP} as a black-box for analyzing the accuracy.  The full algorithm is summarized below.

\begin{algorithm}[H]
\caption{{\sc Full Robust $\Z_2$-Synchronization} }
\begin{algorithmic}
\State \textbf{Input:} Matrix $A \in \R^{n \times n}$, parameters $\lambda, \eps$
\State Run {\sc Compute Initial Labelling ($\Z_2$-Synchronization)} to compute labelling $\ell_{\text{init}} \in \{-1,1 \}^n$
\If{$\eps \geq 1/\lambda$}
\State \textbf{Output:} $\ell_{\text{init}}$
\Else 
\State Run {Boosting Using SDP} on $A, \ell_{\text{init}}$ and parameters
\begin{align*}
d & \leftarrow  \sqrt{n} \\
\zeta & \leftarrow \frac{ 10^5}{\lambda} \\  
K & \leftarrow  10^6
\end{align*}
\State \textbf{Output:} final labelling $\ell$
\EndIf
\end{algorithmic}
\label{alg:full-Z2}
\end{algorithm}
\begin{proof}[Proof of Theorem~\ref{thm:main-sync}]
By Lemma~\ref{lem:Z2-rough-clustering}, with probability $ 1 - 2/n^2$, the accuracy of the initial labelling is at least 
\[
1 - 10^3(1/\lambda + \eps) \,.
\]
Thus, in the case where $\eps \geq 1/\lambda$, we are immediately done.  Otherwise, the above accuracy is at least $1 - 0.1\zeta$.  Now it suffices to verify the remaining conditions of Theorem~\ref{thm:boosting-SDP}.  Let $\wt{\ell}$ denote the true labelling and let $\wt{L} = \wt{\ell}\wt{\ell}^T$.  Let $\kappa$ be a parameter that will allow us to balance the accuracy with the failure probability from the generative model.  Set 
\[
\gamma = e^{-\lambda^2 /2 + 3\kappa + (10K)^4\lambda} + \eps + \frac{10^4}{n} \max\left(0, \kappa - \lambda  \right) \,.
\]
In the generation of the matrix $A$, let $A_0$ be matrix  with only the Gaussian noise (before the semi-random noise or adversarial corruptions).  By Claim~\ref{claim:Z2-spectral}, with probability at least $1 - 1/n^5$, 
\[
\norm{\left(A_0 - \lambda \wt{L}/\sqrt{n} \right)}_{\op} \leq 3\sqrt{n} \,.
\]
Now let $S$ be the subset of uncorrupted nodes (after the adversary makes the $\eps$-corruption).  Note that clearly $|S| \geq (1 - \gamma )n $ by the earlier definition of $\gamma$.  By definition, we can write
\[
A = A_0 + F = F + (A_0 - \lambda \wt{L}/\sqrt{n}) +   \lambda \wt{L}/\sqrt{n}
\]
where $F$ is a matrix such that $(F \odot \wt{L})_{S \times S}$ is entry-wise nonnegative.  We can now set
\begin{align*}
Y &= \left( A_0 - \lambda \wt{L}/\sqrt{n} \right)_{S \times S} \\
Z &= \left( \lambda \wt{L}/\sqrt{n}\right)_{S \times S} \,.
\end{align*}
This completes the verification of the second and third properties that we need to apply Theorem~\ref{thm:boosting-SDP}.  

Now, it remains to verify the final property in order to apply Theorem~\ref{thm:boosting-SDP}.  To do this, we rely on Corollary~\ref{coro:Z2-resolvable}.  Setting 
\[
\theta = e^{-\lambda^2/2 + 3\kappa + (10K)^3\lambda}
\]
in Corollary~\ref{coro:Z2-resolvable} immediately implies that for any subset $T \subset S$ with $|T| \geq (1 - \zeta)n$, the matrix $((A - F) \odot \wt{L})_{T \times T} $ is resolvable with parameters 
\[
\left( 100K^3 \sqrt{n},  1.1 \sqrt{n} \left( \theta + \frac{n - |T|}{n} + \frac{10^4}{n} \max(0, \kappa - \lambda) \right)  \right)
\]
with probability at least $1 - e^{-10\kappa} - 2/n^2$. Substituting in the definitions of $d$ and $\gamma$ completes the verification of the last property that we need in order to apply Theorem~\ref{thm:boosting-SDP}.

We conclude that the accuracy of the final labelling $\ell$ is at least $1 - 8\gamma$ with probability at least $1 - e^{-10\kappa} - 4/n^2$.  Finally, to complete the proof and bound the expected accuracy, it suffices to substitute in the expression for $\gamma$ and integrate over the failure probability (which is controlled by $\kappa$) and we get that the expected accuracy is at least 
\[
1 -  8\eps - e^{-\lambda^2/2  + (20K)^3 \lambda}  - \frac{e^{-\lambda}}{n} \geq 1 -  8\eps - e^{-\lambda^2/2  + O(\lambda)}  - \frac{e^{-\sqrt{\log n}}}{n} \,.
\]
This completes the proof.
\end{proof}